\RequirePackage{fix-cm}
\documentclass{svjour3}                     
\smartqed  
\usepackage{graphicx}
 \usepackage{mathptmx}      
\usepackage{latexsym}
\usepackage{amssymb,amsmath}
\usepackage{hyperref}
\usepackage{enumerate}
\begin{document}

\title{\LARGE Can one design a geometry engine? \\
}
\subtitle{ On the (un)decidability of certain affine Euclidean geometries}

\titlerunning{Decidability of geometry}        

\author{Johann A. Makowsky        
}

\authorrunning{J.A. Makowsky} 

\institute{J.A. Makowsky \at
              Faculty of Computer Science, Technion--Israel Institute of Technology, Haifa, Israel \\
              \email{janos@cs.technion.ac.il}           
}
\date{Received: \today / Accepted: date}

\maketitle
\setcounter{tocdepth}{2}

\begin{abstract}
We survey the status of decidabilty of the consequence relation in various
axiomatizations of Euclidean geometry.
We draw attention to a widely overlooked result by Martin Ziegler from 1980,
which proves Tarski's conjecture on the undecidability of finitely axiomatizable theories of fields.
We elaborate on how to use Ziegler's theorem to show that the consequence relations
for the first order theory of the Hilbert plane and the Euclidean plane are undecidable.
As new results we add:
\begin{description}
\item[(A)]
The first order consequence relations for Wu's orthogonal and metric geometries
(Wen-Ts\"un Wu, 1984),
and for the axiomatization of Origami geometry 
(J. Justin 1986, H. Huzita 1991)
are undecidable.
\end{description}
It was already known that  the universal theory of Hilbert planes and 
Wu's orthogonal geometry is decidable.
We show here using elementary model theoretic tools that
\begin{description}
\item[(B)]
the universal first order consequences of any geometric theory $T$ of Pappian planes
which is consistent with the analytic geometry of the reals
is decidable.
\end{description}
The techniques used were all known to experts in mathematical logic and geometry in the 
past  but no detailed proofs are easily accessible for practitioners of 
symbolic computation or automated theorem proving.

\end{abstract}

\newif\ifskip
\skiptrue

\ifskip\else
\newtheorem{definition}{Definition}
\newtheorem{definitions}[definition]{Definitions}
\newtheorem{notation}[definition]{Notation}
\newtheorem{lemma}[definition]{Lemma}
\newtheorem{theorem}[definition]{Theorem}
\newtheorem{proposition}[definition]{Proposition}
\newtheorem{obs}[definition]{Observation}
\newtheorem{corollary}[definition]{Corollary}
\newtheorem{fact}[definition]{Fact}
\newtheorem{examples}[definition]{Examples}
\newtheorem{result}[definition]{Result}
\newtheorem{propex}[definition]{Proposition-Exercise}
\newtheorem{o_Problem}[definition]{Hard Problem}
\newtheorem{observation}[definition]{Observation}
\newtheorem{Digression}{Digression}
\newtheorem{Dtheorem}[Digression]{Theorem}
\newtheorem{Dconjecture}[Digression]{Conjecture}
\newtheorem{Dproposition}[Digression]{Proposition}
\newtheorem{Dexample}[Digression]{Example}
\newtheorem{Dproblem}[Digression]{Problem}
\newcommand{\nwl}{\vspace{5mm}\newline}
\newcommand{\nwp}{\newpage\ \newline} 

\newenvironment{renumerate}{\begin{enumerate}}{\end{enumerate}}
\renewcommand{\theenumi}{\roman{enumi}}
\renewcommand{\labelenumi}{(\roman{enumi})}
\renewcommand{\labelenumii}{(\roman{enumi}.\alph{enumii})}
\newtheorem{exercises}[definition]{Exercises}
\newtheorem{convention}[definition]{Convention}
\newtheorem{principle}[definition]{Principle}
\newtheorem{hint}[definition]{Hint}
\fi 

\newcommand{\N}{{\mathbb N}}
\newcommand{\bE}{{\mathbf E}}
\newcommand{\bO}{{\mathbf O}}
\newcommand{\bV}{{\mathbf V}}
\newcommand{\bN}{{\mathbf N}}
\newcommand{\bR}{{\mathbf R}}
\newcommand{\HF}{\mbox{\bf HF}}
\newcommand{\CNF}{\mbox{\bf CNF}}
\newcommand{\PNF}{\mbox{\bf PNF}}
\newcommand{\QF}{\mbox{\bf QF}}
\newcommand{\DNF}{\mbox{\bf DNF}}
\newcommand{\DISJ}{\mbox{\bf DISJ}}
\newcommand{\CONJ}{\mbox{\bf CONJ}}
\newcommand{\Ass}{\mbox{Ass}}
\newcommand{\Support}{\mbox{Support}}
\newcommand{\V}{\mbox{\bf Var}}
\newcommand{\fA}{{\mathfrak A}}
\newcommand{\fB}{{\mathfrak B}}
\newcommand{\fN}{{\mathfrak N}}
\newcommand{\fZ}{{\mathfrak Z}}
\newcommand{\fQ}{{\mathfrak Q}}
\newcommand{\Aa}{{\mathfrak A}}
\newcommand{\Bb}{{\mathfrak B}}
\newcommand{\Cc}{{\mathfrak C}}
\newcommand{\Gg}{{\mathfrak G}}
\newcommand{\Ww}{{\mathfrak W}}
\newcommand{\Rr}{{\mathfrak R}}
\newcommand{\Nn}{{\mathfrak N}}
\newcommand{\Zz}{{\mathfrak Z}}
\newcommand{\Qq}{{\mathfrak Q}}
\newcommand{\F}{{\mathbf F}}
\newcommand{\T}{{\mathbf T}}
\newcommand{\Z}{{\mathbb Z}}
\newcommand{\R}{{\mathbb R}}
\newcommand{\C}{{\mathbb C}}
\newcommand{\Q}{{\mathbb Q}}
\newcommand{\bP}{{\mathbf P}}
\newcommand{\bQ}{{\mathbf Q}}
\newcommand{\bFP}{{\mathbf{FP}}}
\newcommand{\bNP}{{\mathbf{NP}}}
\newcommand{\MT}{\mbox{MT}}
\newcommand{\TT}{\mbox{TT}}
\newcommand{\card}{\mathrm{card}}
\newcommand{\SOLEVAL}{{\mathrm{SOLEVAL}}}
\newcommand{\MSOLEVAL}{{\mathrm{MSOLEVAL}}}
\newcommand{\CMSOLEVAL}{{\mathrm{CMSOLEVAL}}}
\newcommand{\MSOL}{{\mathrm{MSOL}}}
\newcommand{\CMSOL}{{\mathrm{CMSOL}}}
\newcommand{\CSOL}{{\mathrm{CSOL}}}
\newcommand{\CFOL}{{\mathrm{CFOL}}}
\newcommand{\FOL}{{\mathrm{FOL}}}
\newcommand{\SOL}{{\mathrm{SOL}}}
\newcommand{\cL}{{\mathcal{L}}}
\newcommand{\cF}{{\mathcal{F}}}
\newcommand{\cS}{{\mathcal{S}}}
\newcommand{\cO}{{\mathcal{O}}}
\newcommand{\fF}{{\mathfrak{F}}}
\newcommand{\fH}{{\mathfrak{H}}}
\newcommand{\fU}{{\mathfrak{U}}}
\newcommand{\fE}{{\mathfrak{E}}}
\newcommand{\cA}{{\mathcal{A}}}
\newcommand{\cB}{{\mathcal{B}}}
\newcommand{\cR}{{\mathcal{R}}}
\newcommand{\cH}{{\mathcal{H}}}
\newcommand{\cZ}{{\mathcal{Z}}}
\newcommand{\ATM}{\mathrm{ATM}}
\newcommand{\STM}{\mathrm{STM}}
\newcommand{\GTG}{\mathrm{GTG}}
\newcommand{\GTC}{\mathrm{GTC}}
\newcommand{\cT}{\mathcal{T}}
\newcommand{\WFF}{{\mathrm{WFF}}}

\newcommand{\TFOL}{\mbox{\bf TFOL}}
\newcommand{\FOF}{\mbox{\bf FOF}}
\newcommand{\NNF}{\mbox{\bf NNF}}
\newcommand{\Var}{\mathrm{\bf Var}}
\newcommand{\Term}{\mathrm{\bf Term}}

\newpage
\tableofcontents
\section{Introduction}
\label{se:intro}


\begin{quote}
\raggedleft
{\em ``Truth cannot be demonstrated, only invented.''}
\\
\footnotesize
The bishop in Max Frisch's play\\
{\em Don Juan or the love of geometry}, Act V.
\end{quote}

Since the beginning of the computer era automated theorem proving
in geometry remained a central topic and challenge for artificial intelligence.
Already in the late 1950s, \cite{gelernter1959realization,gelernter1960empirical},
H. Gelernter presented a machine implementation of a theorem prover for Euclidean Geometry.

The very first idea for mechanizing theorem proving in Euclidean geometry
came from the fact that till not long ago high-school students were rather proficient
in proving theorems in planimetry using Euclidean style deductions.
A modern treatment of Euclidean Geometry was initiated by D. Hilbert at the 
end of the 19th century \cite{hilbert1902foundations}, and a modern reevaluation of Euclidean Geometry
can be found in \cite{bk:Hartshorne2000}.
Formalization in first order logic is thoroughly discussed in \cite{avigad2009formal},
and conceptual issues are discussed in
\cite[Part III]{bk:Baldwin2018}.

On the high-school level one distinguishes between {\em Analytic Geometry} which is the geometry
using coordinates  ranging over the real numbers, and {\em Synthetic Geometry} which deals
with points and lines with their incidence relation augmented by various other relations
such as {\em equidistance, orthogonality, betweenness, congruence of angles} etc.
A geometric statement is, in the most general case, a formula in {\em second order logic $\SOL$}
using these relations. However, it is more likely that for practical purposes
full second order logic is rarely used. 
In fact, all the geometrical theorems proved in \cite{hilbert1902foundations} are expressible 
by formulas of of first order logic $\FOL$ with very few qunatifier alterntaions, cf. also
\cite{avigad2009formal,miller2007euclid}.
Instead, one uses statements expressed in a suitable
fragment $\fF$ of second order logic, which can be full first order logic $\FOL$
(the {\em Restricted Calculus} in the terminology of \cite{bk:HilbertAckermann-eng})
or an even more restricted fragment, 
such as the universal $\forall$-formulas $\fU$,
the existential $\exists$-formulas $\fE$,
or $\forall\exists$-Horn formulas $\fH$
of first order logic.

Many variants of Synthetic Euclidean Geometry are axiomatized in the language of first order logic
by a finite set of axioms or axiom schemes $T \subseteq \FOL$ if {\em continuity requirements are discarded}.
It follows from the Completeness Theorem of first order logic that
the first order consequences of  $T$ are recursively (computably) enumerable.
If full continuity axioms, which not $\FOL$-expressible, are added even the first order consequences
of $T$ are not necessarily recursively enumerable.

A first order statement in the case of Analytic Geometry over the reals
is a first order formula $\phi$ in the language of ordered fields and we ask whether $\phi$
is true in the ordered field of real numbers.
By a celebrated theorem of A. Tarski  announced in \cite{ar:Tarski1931},
and proven in \cite{bk:Tarski1951},
this question is mechanically decidable using quantifier elimination. 
However, the complexity of the decision procedure given
by Tarski uses an exponential blowup for each elimination of a quantifier. This has been  dramatically improved
by G.E. Collins in 1975, reprinted in \cite{caviness2012quantifier},
giving a doubly exponential algorithm in the size of the input formula.
Further progress was and is slow. For a state of the art discussion, cf. \cite{bapr:03,caviness2012quantifier}.
However, it is unlikely that a polynomial time algorithm exists for quantifier elimination over the
ordered field of real numbers for existential formulas, because this would imply that
in the computational model of Blum-Shub-Smale over the reals $\R$, \cite{bk:BCSS},
we would have $\bP_{\R} = \bNP_{\R}$, \cite{bk:Poizat,prunescu2006fast}, which is one of the
open  Millennium Problems,
\cite{carlson2006millennium}.
For formulas with unrestricted alternation of quantifiers a doubly exponential lower bound 
was given in \cite{dahe:88}. 
Simply exponential upper bound for existential formulas 
were given by several authors. For a survey, see \cite{bapr:03}.

So what can a geometry engine for Euclidean Geometry try to achieve?

For a fixed fragment $\fF$ of $\SOL$ in the language of Analytic or Synthetic Geometry
we look at the following possibilities:
\begin{description}
\item[\bf Analytic Tarski Machine $\ATM(\fF)$:] \ 
\begin{description}
\item[Input:]
A first order formula $\phi \in \fF$ in the language of ordered fields.
\item[Output:]
{\bf true} if $\phi$ is true in the ordered field of real numbers,
and {\bf false} otherwise.
\end{description}
\item[\bf Synthetic Tarski Machine $\STM(\fF)$:] \ 
\begin{description}
\item[Input:]
A first order formula $\psi \in \fF$ in a language of synthetic geometry.
\item[Output I:] a translation $\phi = cart(\psi)$ of $\psi$ into the language
of analytic geometry.
\item[Output II:]
{\bf true} if $\phi$ is true in the ordered field of real numbers,
and {\bf false} otherwise.
\end{description}
\item[\bf Geometric Theorem Generator $\GTG(\fF)$:] \ 
\begin{description}
\item[Input:]
A recursive set of first order formulas $\T \subseteq \FOL $ (not necessarily in $\fF$), 
in the language of some synthetic  geometry.
\item[Output:] A non-terminating sequence $\phi_i: \in \N$ 
of formulas in $\fF$ the language of the same synthetic  geometry which are consequences of $\T$.
\end{description}
\item[\bf Geometric Theorem Checker $\GTC(\fF)$:] \ 
\begin{description}
\item[Input:]
A recursive set of first order formulas $\T$ and another formula $\phi \in \fF$ 
in the language of some synthetic  geometry.
\item[Output:] 
{\bf true} if $\phi$ is a consequence of $\T$,
and {\bf false} otherwise.
\end{description}
\end{description}

In the light of the complexity of quantifier elimination over the real numbers, \cite{bk:Poizat,prunescu2006fast},
designing  {\em computationally feasible} 
Analytic or Synthetic Tarski Machines for various fragments $\fF$ with the exception of $\fU$ is
a challenge both for Automated Theorem Proving (ATP) as well as for Symbolic Computation (SymbComp).
Designing  
Geometric Theorem Generators $\GTG(\fF)$
is possible but
seems pointless, 
because it will always output long sub-sequences of geometric theorems in which we
are not interested.

In this paper we will concentrate on the challenge of designing
Geometric Theorem Checkers $\GTC(\fF)$. 
This is possible only for very restricted fragments $\fF$ of $\FOL$, such as the universal formulas $\fU$.

The main purpose of this paper is to bring {\em negative results} 
concerning Geometric Theorem Checkers $\GTC(\fF)$
to a wider audience.

The negative results are based on a  correspondence between sufficiently strong
axiomatizations of Synthetic Euclidean Geometries and certain theories of fields consistent
with the theory of the ordered field of real numbers.

A model of incidence geometry is an incidence structure which satisfies the axioms I-1, I-2 and I-3
from Section \ref{se:incidence}.
An {\em affine plane} is a model
of incidence geometry satisfying the Parallel Axiom (ParAx).
An affine plane is {\em Pappian} is it additionally satisfies the axiom of Pappus (Pappus).
In this paper an axiomatization of geometry $T$ is
sufficiently strong if all its models are affine planes.

Let $\cF$ be a field of characteristic $0$.
One can construct an Cartesian plane 
%
$\Pi(\cF)$ over $\cF$
which satisfies the Pappian axiom, and where all the lines are infinite.
This construction is an example of a transduction as defined in Section \ref{se:transductions}.
On the other side, if $\Pi$ is an Pappian plane which has no finite lines 
then one can define inside $\Pi$ its coordinate field $\cF(\Pi)$ which is of characteristic $0$.

\begin{proposition}[F. Schur \cite{schur1909grundlagen} and E. Artin \cite{artin1957geometric}]
\label{th:Artin}
\label{th:onto}
\begin{enumerate}[(i)]
\item
$\cF$ is a field of characteristic $0$ iff $\Pi(\cF))$ is a 
Pappian plane with no finite lines.
\item
$\Pi$ is a 
Pappian plane with no finite lines
iff
$\cF(\Pi))$ is a field of characteristic $0$. 
\item
The fields $\cF$ and $\cF(\Pi(\cF))$ are isomorphic.
\item
The Pappian planes 
$\Pi$ and $\Pi(\cF(\Pi))$ are isomorphic as incidence structures.
\end{enumerate}
\end{proposition}

A theory (set of formulas) $T \subseteq \FOL(\tau)$
is  {\em axiomatizable} 
if the set of consequences of $T$ 
is computably enumerable.
$T$ is {\em decidable}
if the set of consequences of $T$ 
is computable.
$T$ is {\em undecidable} if it is not decidable.
$T$ is {\em complete}
if for every formula $\phi \in \FOL(\tau)$ without free variables either 
$T \models \phi$ or
$T \models \neg \phi$.
We note that if $T$ is axiomatizable and complete, the $T$ is decidable.

On the side of theories of fields we have several undecidability results:

\begin{proposition}[J. Robinson, 1949 \cite{ar:JRobinson49}]
\label{th:JRobinson}
\begin{enumerate}[(i)]
\item
The theory of fields is undecidable.
The same holds for fields of characteristic $0$.
\item
The theory of ordered fields is undecidable.
\item
The theory of the field of rational numbers $\langle \Q, +, \times \rangle$
is undecidable.
\end{enumerate}
\end{proposition}

To show that the first order theory of affine geometry is undecidable we would like to use a classical tool from
decidability theory, the details of which we explain in Section \ref{se:undecidability}.

\begin{proposition}[\cite{Rabin65}, based on \cite{bk:TarskiMostowskiRobinson53}]
\label{th:undecidable}
Let $I$ be a first order translation scheme with associated transduction $I^*$
which maps $\tau$-structures into $\sigma$-structures.
Furthermore,
let $S$ be an undecidable first order theory over a relational vocabulary $\sigma$
and let $T$ be a theory over $\tau$.
Assume that
$I^*$ maps the models of $T$ {\bf onto} the models of $S$, and that $S$ is undecidable,
then $T$ is also undecidable.
\end{proposition}
The onto-condition needed for our purpose is rarely stated in textbooks. However, 
it is explicitely stated in \cite{bk:Hodges93}.

Propositions \ref{th:Artin} and  \ref{th:JRobinson} are not enough to prove that 
first order theory of affine geometry is undecidable.
We have to verify all the conditions of Proposition \ref{th:undecidable}.

In particular, we have to show:

\begin{enumerate}[(A)]
\item
There is a first order translation scheme $RF_{field}$
such that for every Pappian plane $\Pi$ the structure $RF_{field}^*(\Pi)$ is a field.
\item
There is a first order translation scheme $PP_{\in}$
such that for every field $\cF$ the structure $PP_{\in}^*(\Pi)$ is an
Pappian plane. 
\item
For every field $\cF$ we have 
$$
RF_{field}^*(PP_{\in}^*(\cF)) \simeq  \cF.
$$
\end{enumerate}
All this is shown in detail in Section \ref{se:ptr}.
While the existence of $PP_{\in}$ is rather straightforward, the existence of $RF_{field}$
with the necessary properties (B) and (C)
requires the first order definability of the coordinatization of affine planes.
If $\Pi$ is a Hilbert plane or a Euclidean plane, coordinatization can be achieved through
segment arithmetic, which can be achieved via a first order translation scheme $FF_{field}$,
which is somehow simpler that $RR_{field}$.

Only after having established (A) and (C) we can conclude:
\begin{theorem}
\label{th:affine-undec}
\begin{enumerate}[(i)]
\item
The first order theory of Pappian planes undecidable.
\item
The first order theory of affine geometry is undecidable.
\end{enumerate}
(ii) follows from (i) because Pappian planes are obtained from affine planes by adding a finite
number of axioms in the language of incidence geometry.
\end{theorem}
The ingredients for proving
Theorem \ref{th:affine-undec} were all implicitly available when Proposition \ref{th:JRobinson}
was published. I would also assume that Theorem \ref{th:affine-undec} was known in Berkeley,
but no detailed proof was written down.
A. Tarski presented the result for projective planes at the 11th Meeting of the Association of Symbolic
Logic already in 1949, \cite{tarski-asl-1949}.
An incomplete sketch of a proof 
Theorem \ref{th:affine-undec}
was published in 1961 by 
W. Rautenberg \cite{rautenberg1961unentscheidbarkeit}.
His more detailed proof of the projective case in \cite{rautenberg1962unentscheidbarkeit}
uses Proposition \ref{th:JRobinson} and Lemma \ref{le:interpretability}, but fails to note that something like
Theorem \ref{th:szmielew} is needed to complete the argument. We discuss this in detail 
at the end of Section \ref{se:undecidability}.
It also seems that W. Szmielew planned to include a proof of Theorem \ref{th:affine-undec} 
in her unfinished and posthumously published \cite{szmielew1983affine}.
The only complete proof of Theorem \ref{th:affine-undec} I could find in the literature
appears in \cite{balbiani2007logical}. However, the arguments contain some fixable errors\footnote{
In \cite[Section 7]{balbiani2007logical} the undecidability of affine and projective spaces
is stated (Corollary 7.38). It also discusses Proposition \ref{th:Artin} in  \cite[Section 6]{balbiani2007logical},
but fails to mention that Proposition \ref{th:Artin} is needed
to prove Theorem \ref{th:affine-undec} (Theorem 1.37 in \cite[Section 7]{balbiani2007logical}).
It also attributes Proposition \ref{th:JRobinson} erroneously  to A. Tarski.
}.
One of the purposes of this paper is to give a conceptually clear account of what is needed to prove
Theorem \ref{th:affine-undec}.

To repeat this argument for other axiomatizations of extensions of affine geometry we
need the following theorem of M. Ziegler: 

\begin{theorem}[M. Ziegler, 1982 \cite{ar:ziegler,BeesonZiegler}]
\label{th:Ziegler}
\begin{enumerate}[(i)]
\item
Let $T$ be a finite subset of the theory of the reals $\langle \R, +, \times \rangle$
and let $T^* = T \cup \{ \mathbf{n} \neq 0, n \in \N \}$, where $\mathbf{n}$
is shorthand for $\underbrace{1+ \ldots + 1}_{n}$.
Both $T$ and $T^*$ are undecidable.
\item
The same holds if $T$ is a finite subset of the theory of the complex numbers $\langle \C, +, \times \rangle$.
\end{enumerate}
\end{theorem}
We paraphrase this theorem, following \cite{shlapentokh2014definability},
by saying that the theory of real closed (algebraically closed) fields
of characteristic $0$ is {\em finitely hereditarily undecidable}.

Theorem \ref{th:Ziegler} was conjectured\footnote{
In \cite{hauschild1974rekursive} a proof was announced, which later was found containing in irreparable
mistake, cf. \cite{ar:ziegler}.  }
by A. Tarski, but only proved in 1982 by M. Ziegler.
Ziegler's Theorem remained virtually unnoticed, having been published in German in the Festschrift in honor of
Ernst Specker's 60th birthday, published as a special issue of {\em L'Enseignement Math\'ematique}.
In \cite{bk:Schwabhaeuser1983} the significance of the results of \cite{ar:ziegler} is recognized. 
However, the book is written in German and is usually quoted for its presentation of Tarskian geometry.
The discussion of Theorem \ref{th:Ziegler} is buried there in the second part of the book
dealing with metamathematical questions of geometry.
This part of the book is difficult to absorb, both because of its pedantic style
and its length. 
In short, the only reference to Theorem \ref{th:Ziegler} within the the framework of ATP and SC
is \cite{beeson2012proof}. 
A very short and casual
mention of Theorem \ref{th:Ziegler} 
can also be found in \cite{balbiani2007logical}.

The present paper gives a survey on the status of decidability of various axiomatizations
of Euclidean Geometry, including Wu's metric geometry and  the  Origami geometry which are all
undecidable, see Theorems 
\ref{th:euclid-undec},
\ref{th:wu-undecidable} and
\ref{th:undec-origami}.
None  of these theorems are technically new. They all could have been proven with
the tools used in the proof of Theorem \ref{th:affine-undec} together with Ziegler's Theorem \ref{th:Ziegler}.
However, Theorem \ref{th:euclid-undec} is stated and proved only in \cite{bk:Schwabhaeuser1983},
and Theorems
\ref{th:wu-undecidable} and
\ref{th:undec-origami} could not have been stated before the corresponding geometries
were axiomatized. For Wu's orthogonal geometry this would be 1984 respectively 1994 , when
the first translation from Chinese appeared \cite{bk:Wu1994}, or 1986 \cite{wen1986basic}.
For Origami geometry this would be at the earliest in 1989, \cite{justin1989resolution},
but rather in 2000 with \cite{alperin2000mathematical}.

The purpose of this paper is to discuss undecidability results in geometry
addressing practitioners in Automated Theorem Proving, Articial Intelligence, 
and Symbolic Computation.
Although many variants of these results were stated and understood already in the early 1950s,
I could not find references with detailed proofs which could be easily understood and reconstructed
by graduate students of Logic in Computer Science. 
On the other hand the techniques described in this paper are well known in the 
mathematical logic community. Theorem \ref{th:affine-undec}
and some of its variations are given as an exrecise in \cite[Exercise 10 of Section 5.4]{bk:Hodges93}.
Although the Theorems \ref{th:wu-undecidable} and \ref{th:undec-origami} are strictly speaking new, their proofs  use the same techniques,
together with Ziegler's Theorem \ref{th:Ziegler} from 1982.

We hope that our 
presentation of this material 
is sufficiently  concise and transparent 
in showing the limitations of automatizing theorem proving in affine geometry.
We restrict our discussion here to theories of affine Euclidean geometries.
However, the methods can be extended to projective and hyperbolic geometries.

\subsection*{Outline of the paper}

In Section \ref{se:fdecidability} we summarize what is known about the
(un-)decidability of theories of fields. 
Theorems \ref{th:Ziegler} and \ref{th:macintyre} show
that the decidability of the theory of real closed
fields and its elimination of quantifiers are very specific properties of this theory.

In Section \ref{se:axioms} describe the geometrical theories which are the center
of our discussion: Affine incidence geometry, Hilbert-style Euclidean geometry, Wu's orthogonal
geometry and Origami geometry.

In Section \ref{se:undecidability} we spell out the subtleties needed to derive undecidability of geometrical
theories from the undecidability of corresponding theories of fields.
Although the general idea is very intuitive, the argument given frequently in the literature
tends to overlook that this reduction depends on deep theorems specific to geometry.
Besides the one-one correspondence between geometrical theories and theories of fields one also needs
the first order definability of the coordinatization theorem for affine incidence geometry.
In Section \ref{se:coord} we do discuss the role coordinatizations play in the undecidability proofs.
and show that coordinatization is first order definable.
In Section \ref{se:undecidable} we finally give the complete proofs of undecidability
of our geometrical theories, and in Section \ref{se:decidable} we 
show that the consequences in the universal fragment $\fU$
of these geometries are still decidable.
In Section \ref{se:conclu} we summarize what we have achieved and propose some open problems.

\subsection*{An after-thought concerning computer-checkable proofs}
As some authors misquote Proposition \ref{th:undecidable} by omitting the condition that
$I^*$ has to  map the models of $T$ {\bf onto} the models of $S$, it would be interesting to
know whether a proof-checking system would have helped in discoving the exact nature of the gap
in the published incomplete proofs of Theorem \ref{th:affine-undec}.

\section{Decidable and undecidable theories of fields}
\label{se:fdecidability}
\subsection{Background on fields}
Let $\tau_{field}$ be the purely relational vocabulary consisting of a 
ternary relation $Add(x,y,z)$ for addition with $Add(x,y,z)$ holds if $x+y=z$,
a ternary relation $Mult(x,y,z)$ for multiplication with $Mult(x,y,z)$ holds if $x\cdot y=z$,
and two constants for the neutral elements $0$ and $1$. 
A field $\cF = \langle A, Add_A, Mult_A, 0_A, 1_A \rangle $ is a $\tau_{field}$-structure
satisfying the usual field axioms, which we write for convenience in the usual notation with
$+$ and $\cdot$.
Let $\tau_{ofield}$ be the purely relational vocabulary $\tau_{field} \cup \{\leq\}$
where $\leq$ is a binary relation symbol.
An ordered field $\cF = \langle A, Add_A, Mult_A, 0_A, 1_A, \leq_A \rangle $ is a $\tau_{ofield}$-structure
satisfying the usual axioms of ordered fields.

We sometimes also look at (ordered) fields as structures over a vocabulary containing
function symbols.
Let
$\tau_{f-field}$ 
be the vocabularies
with
binary functions for addition and multiplication,
unary functions for negatives $-x$ and inverses $\frac{1}{x}$,
and
$\tau_{f-ofields} = \tau_{f-field} \cup \{\leq\}$.

The difference between the relational and functional version
lies in the notion of substructure. 
In the functional version substructures of (ordered fields are (ordered) fields.
Formulas in the functional version can be translated into formulas in the relational
version but this requires the use of existential quantifiers.

Let $B(x_1, \ldots, x_m, \bar{y})$
be a quantifier free formula with free variables $x_1, \ldots, x_m, \bar{y}$.
A formula $\phi$ 
with free variables $\bar{y}$
is {\em universal} if it is of the form
$$
\phi = \forall x_1, \ldots, \forall x_m B(x_1, \ldots, x_m, \bar{y}).
$$
A formula $\psi$ is
{\em existential} if it is of the form
$$
\psi = \exists x_1, \ldots, \exists x_m B(x_1, \ldots, x_m, \bar{y})
$$
Note that when translating a quantifier-free formula in $\FOL_{f-field}$
into an equivalent formula in $\FOL_{field}$, the result is not quantifier-free but in general 
an existential formula. Translating a universal formula results in an $\forall\exists$-formula.

Let $\cF$ be a field.
\begin{enumerate}[(i)]
\item
For $p$ a prime,
$\cF$ is of {\em characteristic $p$} if $\underbrace{1+ \ldots +1}_{p} = 0$.
\item
$\cF$ is of {\em characteristic $0$} if for all $n \in \N$ we have
that $\underbrace{1+ \ldots +1}_{n} \neq 0$.
\item
$\cF$ is {\em Pythagorean} if every sum of two squares is a square,
$$
\forall x \forall y \forall z (x = y^2 + z^2 \rightarrow \exists u (u^2 =x)).
$$
\item
$\cF$ is a {\em Vieta field} if every polynomial with coefficients in $\cF$ of degree at most $3$ has
a root in $\cO$.
\item
$\cF$ is {\em formally real} if $0$  cannot be written as a sum of nonzero squares,
i.e., for all $n \in \N$ we have
$$
\forall x_1, \ldots , x_n 
(
\sum_{i=1}^n x_i^2 = 0
\rightarrow
\bigwedge_{i=1}^n (x_i^2 = 0) 
)
$$
\item
$\cF$ is {\em algebraically closed} if
every non-constant polynomial 
with coefficients in $\cF$ has a root in $\cF$. 
We denote by $ACF_0$ the first order sentences of fields describing an algebraically closed field
of characteristic $0$.
\end{enumerate}

An ordered field $\cO$ is a field $\cF$ with an additional binary relation $\leq$
which is compatible with the arithmetic relations of $\cF$.
An ordered field is always of characteristic $0$.

Let $\cO$ be an ordered field.
\begin{enumerate}[(i)]
\item
$\cO$ is {\em Euclidean} if every positive element has a square root,
$$
\forall x ( x \geq 0 \rightarrow \exists y (y^2 =x)).
$$
\item
An ordered field is Pythagorean (Vieta, formally real) if it is an ordered field and
as a field is Pythagorean (Vieta, formally real).
\item
An ordered field
$\cO$ is {\em real closed}  if $\cO$ is formally real,
every positive element in $\cO$ has a square root, 
$$
\forall x \exists y (y^2=x)
$$
and every polynomial of odd degree
with coefficients in $\cO$ has a root in $\cO$. 
$$
\forall x_0, \ldots x_{2n+1} ( x_{2n+1} \neq 0) \rightarrow
\exists y
\sum_{i=0}^{2n+1} x_i y^i =0
)
$$
We denote by $RCF$ the first order sentences of ordered fields describing a real closed field.
\end{enumerate}

\subsection{Undecidable theories of fields}
We now are ready to  apply Ziegler's Theorem (Theorem \ref{th:Ziegler}) in order to show the following:

\begin{theorem}
Let $T$ one of the first order theories over the vocabulary of (ordered) fields
listed below.
Then the set of first order consequences of $T$ is undecidable 
(not computable but computably enumerable).
\begin{enumerate}[(i)]
\item
The theory of  fields and of ordered fields.
\item
The theory of Pythagorean fields and ordered Pythagorean fields.
\item
The theory of Vieta fields and ordered Vieta fields.
\item
The theory of Pythagorean fields and ordered Pythagorean fields
of characteristic $0$.
\item
The theory of ordered Euclidean fields.
\end{enumerate}
\end{theorem}
\begin{proof}
First we note that each of these theories has the field of (ordered) real numbers as a model.
Furthermore each of them is either finite, or of the form
$$
T^* = T \cup \{ \mathbf{n} \neq 0, n \in \N \}
$$ 
with $T$ finite.
Hence we can apply Theorem \ref{th:Ziegler}.
\hfill $\Box$
\end{proof}

\subsection{Decidable theories of fields}
In order to prove decidability of the theory Tarskian Geometry
A. Tarski (and A. Seidenberg independently) proved the following:

\begin{proposition}[A. Tarski and A. Seidenberg \cite{basu2014algorithms}]
\label{th:seidenberg}
The first order theory  $RCF \subseteq \FOL_{f-ofield}$ 
is recursiveley axiomatized, complete and admits elimination of quantifiers, and therefore is decidable.
\end{proposition}

A first order theory $T \subseteq \FOL(\tau)$ over some vocabulary $\tau$ is complete
if $T$ is satisfiable, and
for every formula  $\phi \in \FOL(\tau)$ without free variables we have either $T \models \phi$
or $T \models \neg \phi$.

\begin{proposition}[A. Tarski \cite{bk:Tarski1951}]
\label{th:tarski}
The first order theory $ACF_0 \subseteq \FOL_{f-field}$ 
is recursively axiomatized, complete and
admits elimination of quantifiers, and therefore is decidable.
\end{proposition}

\begin{remark}
To prove decidability one has to prove additionally
in both Propositions \ref{th:seidenberg} and \ref{th:tarski}
that equality and inequality  (and comparison by $\leq$) of constant terms of $\tau_{f-field}$
($\tau_{f-ofield}$) is decidable.
We also note that quantifier elimination is not possible if the theories are expressed
in $\FOL_{ofield}$ respectively $\FOL_{field}$.
\end{remark}

However, even in $\FOL_{f-ofield}$ respectively $\FOL_{f-field}$ the method  
of quantifier elimination cannot be used for other theories 
compatible with the theories $RCF$ or $ACF_0$.

\begin{theorem}[\cite{macintyre1983elimination}]
\label{th:macintyre}
Assume 
$T \subseteq \FOL_{f-field}$ 
($T \subseteq \FOL_{f-ofield}$)
is a theory of (ordered fields) which has the complex (real) numbers as a model,
and $T$ admits elimination of quantifiers, then $T$ is equivalent to 
$ACF_0$ ($RCF$).
\end{theorem}

\begin{problem}
Is there a decidable (infinite) theory $T$ of ordered fields which has no real closure?
\end{problem}

Inside the field of real numbers there exists a minimal Pythagorean  $\bP$ 
(Euclidean $\bE$, Vieta $\bV$)
field, which is the intersection of all Pythagorean (Euclidean, Vieta) subfields  in $\R$.
The theory of the minimal field of characteristic $0$, the field $\bQ$ of the rationals $\Q$ 
is undecidable by Proposition \ref{th:JRobinson}(iii).

\begin{problem}
Are the complete theories of (ordered) fields of $\bP$, $\bE$ or $\bO$ undecidable?
\end{problem}

Theorem \ref{th:Ziegler} holds not only for finite subtheories of real or algebraically closed fields,
of characteristic $0$, but also for finite characteristic and for certain formally $p$-adic fields.
In
\cite{shlapentokh2014definability}, many more infinitely axiomatizable theories of fields are
shown to be finitely hereditarily undecidable.

\subsection{The universal consequences of a theory of fields}
Our next observation concerns the universal consequences of a theory of fields.

The following is a special case
of Tarski's Theorem for universal formulas proven in every
textbook on model theory, e.g., \cite{bk:Hodges93}.
\begin{lemma}
\label{le:substructures}
Let $\cF$ be a field and $\cF_0$ be a subfield.
Let $\theta \in \FOL(\tau_{f-field})$ be a universal formula with parameters from $\cF_0$, 
and $\cF \models \theta$
Then $\cF_0 \models \theta$.
\\
The same also holds for ordered fields.
\end{lemma}

\begin{proposition}
\label{pr:f-universal}
\begin{enumerate}
\item
Let $F$ be a set of $\tau_{f-field}$-sentences
such that all its models are fields of characteristic $0$, and $F$ is consistent with $ACF_0$,
then  for every universal $\theta \in \FOL_{f-field}$
we have:
$$
F \models \theta
\mbox{ iff }
ACF_0 \models \theta.
$$
Hence the universal consequences of $F$ are decidable.
\item
Let $F_o$ be a set of $\tau_{f-ofield}$-sentences
such that all its models are ordered fields, and $F_o$ is consistent with $RCF$,
then  for every universal $\theta \in \FOL_{f-ofield}$
we have:
$$
F_o \models \theta
\mbox{ iff }
RCF \models \theta.
$$
Hence the universal consequences of $F_o$ are decidable.
\end{enumerate}
\end{proposition}
\begin{proof}
(i):
As $ACF_0$ is complete and $F$ is consistent with $ACF_0$ we have that
$ACF_0 \models F$.
Let $F \models \theta$. Then also $ACF_0 \models \theta$.
\\
Conversely, assume $ACF_0 \models \theta$. Now we use that $\theta$ is universal.
By Lemma \ref{le:substructures}, 
$\theta$ holds in every subfield $\cF$ of an algebraically closed fields of characteristic $0$.
By a classical theorem of Algebra, \cite{ar:steinitz},
every field of characteristic $0$ has an algebraically closed extension which satisfies
$ACF_0$. Hence $T \models \theta$.
\\
(ii): The proof is similar, using real closures instead.
\hfill $\Box$
\end{proof}

In \cite{bk:Wu1994} a special case of the decidability in Proposition \ref{pr:f-universal}(i)  
is proved, where the decision procedure is given using Hilbert's Nullstellensatz and Gr\"obner bases,
rather than via quantifier elimination.
We discuss this further in Section \ref{se:decidable}.
This makes the decision procedure seemingly less complicated than 
in the case of the decidability in Proposition \ref{pr:universal}(ii).
A comparison of the complexity of the  two cases may be found in \cite{kapur1988refutational}.

\begin{problem}
\label{problem-existential}
For which theories of fields $F$ is the consequence problem for existential formulas $\fE$
decidable.
\end{problem}

The answer is positive for $ACF_0$ and $RCF$ by Propositions \ref{th:seidenberg} and \ref{th:tarski}.
In spite of recent results by J. Koengismann \cite{koenigsmann2016defining,koenigsmann2016question}
on decidability of theories of fields, Problem \ref{problem-existential} is open 
for the field of rational numbers $\Q$.

\begin{problem}
\label{problem-Q}
Is the existential theory of the field $\langle \Q, +_Q, \times_Q, 0, 1 \rangle$
decidable?
\end{problem}

\section{Axioms of geometry: Hilbert, Wu and Huzita-Justin}
\label{se:axioms}
In this section we collect some of Hilbert's axioms of geometry  which we need in the sequel, and 
which are all true when one considers the analytic geometry  of the plane with real coordinates.
\subsection{The vocabularies of geometry}
\label{se:glanguage}
Models of plane geometry are called {\em planes}.
These models differ in their basic relations.
The universe is always two-sorted, consisting of 
$\mathrm{Points}$
and
$\mathrm{Lines}$
and the most basic relation is {\em incidence $\in$}
with $p \in \ell$ to be interpreted as a point $p$ is coincident with a line $\ell$.
Other relations are
\begin{description}
\item[\bf Equidistant:]  $Eq(p_1,p_2, p_1', p_2')$ 
to be interpreted as two pairs of points $p_1, p_2$ and $p_1', p_2'$
have the same distance.
\item[\bf Orthogonality:] $Or(\ell_1, \ell_2)$ to be interpreted as two lines are orthogonal (perpendicular)/
\item[\bf Equiangular:] $An(p_1, p_2, p_3, p_1', p_2', p_3')$ to be interpreted as two triples of
points define the same angle.
\item[\bf Betweenness:] $Be(p_1, p_2, p_3)$ to be interpreted as three distinct points are on the same line and
$p_2$ is between $p_1$ and $p_3$.
\item[\bf P-equidistant:] $Peq(\ell_1, p, \ell_2)$ to be interpreted as the point $p$ has the same
distance from two lines $\ell_1$ and $\ell_2$.
\item[\bf L-equidistant:] $Leq(p_1, \ell, p_2)$ to be interpreted as the two points $p_1$ and $p_2$
have the same distance from the line $\ell$.
\item[\bf Symmetric Line:] $SymLine(p_1, \ell, p_2)$ to be interpreted as the two points $p_1$ and $p_2$
are symmetric with respect to the line $\ell$.
\end{description}

We define now the following vocabularies:
\begin{description}
\item[$\tau_{\in}$:] The vocabulary of incidence geometry, which uses incidence alone, possibly extended
with a few symbols for specific constants.
\item[$\tau_{hilbert}$:] The vocabulary of Hilbertian style geometry: 
Incidence, Betweenness, Equidistance and Equiangularity \cite{bk:Hartshorne2000}.
\item[$\tau_{wu}$:] The vocabulary of Wu's Orthogonal geometry: Incidence, Equidistance, Orthogonality \cite[Chapter 2]{bk:Wu1994}.
\item[$\tau_{origami}$:] The vocabulary used to describe Origami constructions: 
Incidence, Symmetric Line, L-equidistant, Orthogonality, 
\cite{ghourabi2007logical}
\item[$\tau_{o-origami}$:] An alternative version for describing Origami constructions. 
Incidence, Equidistance, Orthogonality, hence $\tau_{o-origami} = \tau_{wu}$.
\end{description}

We note that all these vocabularies contain the symbol $\in$ for the incidence relation.


In cwthe following subsections we collect some of Hilbert's axioms of geometry  which we need in the sequel, and
which are all true when one considers the analytic geometry  of the plane with real coordinates.

\subsection{Incidence geometries}
\label{se:incidence}

\paragraph{Axioms using only the incidence relation}
\begin{description}
\item[(I-1):]
For any two distinct points $A,B$ there is a unique line $l$
with $A \in l$ and $B\in l$.
\item[(I-2):]
Every line contains at least two distinct points.
\item[(I-3):]
There exists three distinct points $A, B, C$ such that no line $l$
contains all of them.
\end{description}
They can be formulated in $\FOL$ using the incidence relation only.

\paragraph{Parallel axiom}
We define: $Par(l_1, l_2)$ or $l_1 \parallel l_2$
if $l_1$ and $l_2$ have no point in common.

\begin{description}
\item[(ParAx):]
For each point $A$ and each line $l$ there is at most one line $l'$
with $l \parallel l'$ and $A \in l'$.
\end{description}
$Par(l_1, l_2)$ can be formulated in $\FOL$ using the incidence relation only,
hence also the Parallel Axiom.

\paragraph{Pappus' axiom}
\begin{description}
\item[(Pappus):]
Given two lines $l, l'$ and points
$A,B,C \in l$ and $A',B',C' \in l'$
such that
$AC' \parallel A'C$ and 
$BC' \parallel B'C$.
Then also
$AB' \parallel A'B$.
\end{description}

\paragraph{Axioms of Desargues and of infinity}
\begin{description}
\item[(InfLines):]
Given  distinct $A, B, C$ and $l$ with $A \in l, B, C \not\in l$
we define $A_1 = Par(AB,C) \times l$,
and inductively,
$A_{n+1} = Par(A_nB,C) \times l$.
Then all the  $A_i$ are distinct.

Note that this axiom is stronger than just saying there infinitely many points. It says that there are no lines
which have only finitely many points.
\item[(De-1):]
If $AA', BB', CC'$ intersect in one point or are all parallel,
and
$AB \parallel A'B'$ and
$AC \parallel A'C'$
then $BC \parallel B'C'$.
\item[(De-2):]
If
$AB \parallel A'B'$,
$AC \parallel A'C'$ and
$BC \parallel B'C'$
then
$AA', BB', CC'$ are all parallel.
\end{description}
The axiom (InfLies) is not first order definable but consists of an infinite set
of first order formulas with infinitely many new constant symbols for the points $A_i$,
and the incidence relation.
The two Desargues axioms are first order definable using the incidence relation only.
\begin{description}
\item[\bf Affine plane:]
Let $\tau_{\in} \subseteq  \tau$ be a vocabulary of geometry.
A $\tau$-structure $\Pi$ is an {\em (infinite) affine plane} if it satisfies  (I-1, I-2, I-3 and the parallel axiom
(ParAx) and (InfLines).
We denote the set of these axioms by $T_{affine}$
\item[\bf Pappian plane:]
$\Pi$ is a {\em Pappian plane} if additionally it satisfies the Axiom of Pappus (Pappus).
We denote the set of these axioms by $T_{pappus}$
\end{description}

In the literature the definition of affine planes vary. Sometimes the parallel axiom is included, and sometimes not.
We always include the parallel axiom, unless indicated explicitly otherwise.

\subsection{Hilbert style geometries}

\paragraph{Axioms of betweenness} 
\begin{description}
\item[(B-1):]
If $Be(A, B, C)$ then there is $l$ with $A, B, C \in l$.
\item[(B-2):]
For every $A, B$ there is $C$ with $Be(A, B, C)$. 
\item[(B-3:]
For each distinct $A, B, C \in l$ exactly one point of the points $A, B, C$ is between the two others.
\item[(B-4):] (Pasch)
Assume the points $A, B, C$ and $l$ in general position,
i.e.
the three points are not on one line, none of the points is on $l$.
Let $D$ be the point at which $l$ and the line $AB$ intersect.
If $Be(A, D, B)$ there is $D' \in l$ with
$Be(A, D', C)$ or
$Be(B, D', C)$.
\end{description}
The axioms of betweenness are all first order expressible in the language
with incidence relation and the betweenness relation.

\paragraph{Congruence axioms: Equidistance}
We write for $Eq(A, B, C, D)$ the usual $AB \cong CD$.
\begin{description}
\item[(C-0):]
$AB \cong AB \cong BA$.
\item[(C-1):]
Given $A, B, C, C'$, $l$ with $C, C' \in l$
there is a unique $D \in l$ with
$AB \cong CD$ and $B(C, C', D)$ or $B(C, D, C')$.
\item[(C-2):]
If $AB \cong CD$ and $AB \cong EF$ then $CD \cong EF$. 
\item[(C-3):] (Addition)
Given $A, B, C, D, E, F$ with $Be(A, B, C)$ and $Be(D, E, F)$,
if $AB \cong DE$ and $ BC \cong EF$, then $AC \cong DF$.
\end{description}
Note that
(C-1) and (C-3) use the betweenness relation $Be$.
Hence they are first order definable using the incidence, betweenness and equidistance relation.

\paragraph{Congruence axioms: Equiangularity}
We denote by $\vec{AB}$ the directed ray from $A$ to $B$, and
by $\angle(ABC)$ the angle between $\vec{AB}$ and $\vec{BC}$.
For 
the congruence of angles
$An(A,B,C,A',B',C')$
we write 
$\angle(ABC) \cong \angle(A'B'C')$ 
\begin{description}
\item[(C-4):]
Given rays 
$\vec{AB}$, $\vec{AC}$ and $\vec{DE}$
there is a unique ray $\vec{DF}$ with
$\angle(BAC) \cong \angle(EDF)$.
\item[(C-5):]
Congruence of angles is an equivalence relation.
\item[(C-6):](Side-Angle-Side)
Given two triangles $ABC$ and $A'B'C'$
with $AB \cong A'B'$, $AC \cong A'C'$
and $\angle{BAC} \cong \angle{B'A'C'}$
then
$BC \cong B'C'$, 
$\angle{ABC} \cong \angle{A'B'C'}$
and $\angle{ACB} \cong \angle{A'C'B'}$.
\end{description}

\paragraph{Axiom E}
Let $A$ be a point and $BC$ be a line segment.
A circle $\Gamma(A,BC)$ is the set of all points $U$ such that $E(A,U,B,C)$.
A point $D$ is inside the circle $\Gamma(A,BC)$ if there is $U$ with $E(A,U,B,C)$ and $Be(A,D,U)$.
A point $D$ is outside the circle $\Gamma(A,BC)$ if there is $U$ with $E(A,U,B,C)$ and $Be(A,U,D)$.
\begin{description}
\item[(AxE):]
Given two circles $\Gamma, \Delta$ such that $\Gamma$ contains 
at least one point inside, and one point outside $\Delta$, 
then $\Gamma \cap \Delta \neq \emptyset$.
\end{description}

\begin{description}
\item[\bf Hilbert plane:]
Let $\tau$ with $\tau_{hilbert} \subseteq  \tau$ be a vocabulary of geometry.
A $\tau$-structure $\Pi$ is an {\em (infinite) Hilbert plane} if it satisfies  (I-1, I-2, I-3),
(B-1, B-2, B-3, B-4) and
(C-1, C-2, C-3, C-4, C-5, C-6).
\\
We denote the set of these axioms by $T_{hilbert}$
\item[\bf P-Hilbert plane:]
$\Pi$ is a {\em P-Hilbert plane} if it additionally satisfies (ParAx).
\\
We denote the set of these axioms by $T_{p-hilbert}$
\item[\bf Euclidean plane:]
$\Pi$ is a {\em Euclidean plane} if it is a P-Hilbert plane which also satisfies Axiom E.
\\
We denote the set of these axioms by $T_{euclid}$
\end{description}

\subsection{Axioms of orthogonal geometry}

\paragraph{Congruence axioms: Orthogonality}
We denote by $l_1 \perp l_2$ the orthogonality of two lines $Or(l_1, l_2)$.
We call a line $l$ {\em isotropic} if $l \perp l$.
Note that our definitions do not exclude this.
\begin{description}
\item[(O-1):]
$l_1 \perp l_2$ iff $l_2 \perp l_1$.
\item[(O-2):]
Given $O$ and $l_1$, there exists exactly one line $l_2$ with
$l_1 \perp l_2$ and $O \in l_2$.
\item[(O-3):]
$l_1 \perp l_2$ and $l_1 \perp l_3$ 
then $l_2  \parallel l_3$.
\item[(O-4):]
For every $O$ there is an $l$ with
$O \in l$ and $l \not \perp l$.
\item[(O-5):]
The three heights of a triangle intersect in one point.
\end{description}

\paragraph{Axiom of Symmetric Axis and Transposition}
\begin{description}
\item[(AxSymAx):]
Any two intersecting non-isotropic lines have a symmetric axis.
\item[(AxTrans):]
Let $l, l'$ be two non-isotropic lines with
$A,O, B \in l$, 
$AO \cong OB$ 
and $O' \in l'$
there are exactly two points $A', B' \in l'$ such that
$AB \cong A'B' \cong B'A'$ and
$A'O' \cong O'B'$.
\end{description}
The two axioms are equivalent in geometries satisfying the Incidence,
Parallel, Desargues and Orthogonality 
axioms together with the axiom of infinity.

\begin{description}
\item[\bf Orthogonal Wu plane:]
Let $\tau$ with $\tau_{Wu} \subseteq  \tau$ be a vocabulary of geometry.
A $\tau$-structure $\Pi$ is an {\em orthogonal Wu plane} if it satisfies  (I-1, I-2, I-3),
(O-1, O-2, O-3, O-4, O-5), the axiom of infinity (InfLines), (ParAx), 
and the two axioms of Desargues (D-1) and (D-2).
\\
We denote the set of these axioms by $T_{o-wu}$
\item[\bf Metric Wu plane:]
$\Pi$ is a {\em metric Wu plane} if it satisfies additionally the axiom of symmetric axis (AxSymAx)
or, equivalently, the axiom of transposition (AxTrans).
\\
We denote the set of these axioms by $T_{m-wu}$
\end{description}

The axiomatization of orthogonal is due to W. Wu \cite{wen1986basic,bk:Wu1994,wu2007mathematics}, see also
\cite{pambuccian2007orthogonality}.

\subsection{The Origami axioms}

A line which is obtained by folding the paper is called a {\em fold}.
The first six axioms are known as Huzita's axioms. 
Axiom (H-7) was discovered by K. Hatori. 
Jacques Justin and Robert J. Lang also found axiom (H-7), \cite{wiki-huzita-hatori}.
The axioms (H-1)-(H-7) only express closure under folding operations, and do not define
a geometry. To make it into an axiomatization of geometry we have to that these operations
are performed on an affine plane.

We follow here \cite{ghourabi2007logical}.
The  original axioms and their expression as first order formulas in the vocabulary
$\tau_{origami}$  are as follows:
\begin{description}
\item[(H-1):]
Given two points $P_1$ and $P_2$, there is a unique fold (line) that passes through both of them.
$$
\forall P_1, P_2 \exists^{=1} l (P_1 \in l \wedge P_2 \in l)
$$
\item[(H-2):]
Given two points $P_1$ and $P_2$, there is a unique fold (line) that places $P_1$ onto $P_2$.
$$
\forall P_1, P_2 \exists^{=1} l SymLine(P_1, l, P_2)
$$
\item[(H-3):]
Given two lines $l_1$ and $l_2$, there is a fold (line) that places $l_1$ onto $l_2$.
$$
\forall l_1, l_2 \exists k \forall P 
\left(
P \in k \rightarrow Peq(l_1,P, l_2)
\right)
$$
\item[(H-4):]
Given a point $P$ and a line $l_1$, there is a unique fold (line) perpendicular to $l_1$ that passes through point $P$.
$$
\forall P, l \exists^{=1} k \forall P (P \in k \wedge Or(l, k))
$$
\item[(H-5):]
Given two points $P_1$ and $P_2$ and a line $l_1$, there is a fold (line) that places $P_1$ onto $l_1$ 
and passes through $P_2$.
$$
\forall P_1, P_2 l_1 \exists l_2 \forall P (P_2 \in l_2 \wedge \exists P_2 (SymLine(P_1, l_2, P_2) \wedge P_2 \in l_1))
$$
\item[(H-6):]
Given two points $P_1$ and $P_2$ and two lines $l_1$ and $l_2$, there is a fold (line) that places $P_1$ onto $l_1$ and $P_2$ onto $l_2$.
$$
\forall P_1, P_2 l_1, l_2 \exists l_3 
\left( 
(\exists Q_1 SymLine(P_1, l_3, Q_1) \wedge Q_1 \in l_1) \wedge
(\exists Q_2 SymLine(P_2, l_3, Q_2) \wedge Q_2 \in l_2)
\right)
$$
\item[(H-7):]
Given one point $P$ and two lines $l_1$ and $l_2$, there is a fold (line) that places $P$ onto $l_1$ 
and is perpendicular to $l_2$.
$$
\forall P, l_2, l_2 \exists l_3 
\left(
Or(l_2, l_3) \wedge (\exists Q SymLine(P, l_3, Q) \wedge Q \in l_1)
\right)
$$
\end{description}

\begin{description}
\item[\bf Affine Origami plane:]
Let $\tau$ with $\tau_{origami} \subseteq  \tau$ be a vocabulary of geometry.
A $\tau$-structure $\Pi$ is an {\em affine Origami plane} if it satisfies  (I-1, I-2, I-3),
the axiom of infinity (InfLines), (ParAx) and the 
Huzita-Hatori axioms (H-1) - (H-7). 
\\
We denote the set of these axioms by $T_{a-origami}$
\end{description}

\begin{proposition}
The relations $SymLine$ and $Peq$ are first order definable using $Eq$ and  
$Or$ with existential formulas over $\tau_{f-field}$:
Hence the axioms (H-1)-(H-7) are first order definable in $\FOL(\tau_{wu})$.
\end{proposition}
\begin{proof}
\begin{enumerate}[(i)]
\item
$SymLine(P_1, \ell, P_2)$ iff there is a point $Q \in \ell$ such that
$Or((P_1,Q), \ell)$,
$Or((P_2,Q), \ell)$ and 
$Eq(P_1,Q, P_2, Q)$. 
\item
$Peq(\ell_1, P, \ell_2)$ iff there exist points $Q_1, Q_2$ such that 
$Or((P, Q_1), \ell_1)$, 
$Or((P, Q_2), \ell_2)$, 
$Eq(P, Q_1)$ and
$Eq(P, Q_2)$.
\end{enumerate}
\hfill $\Box$
\end{proof}

\section{Proving undecidability of geometrical theories}
\label{se:undecidability}
In this section we spell out how one can apply 
J. Robinson's Proposition \ref{th:JRobinson} or
M. Ziegler's Theorem (Theorem \ref{th:Ziegler})
to prove undecidability of geometric theories.

\subsection{Translation schemes}
\label{se:transductions}
We first introduce the formalism of {\em translation schemes, transductions and
translation}. In \cite{bk:TarskiMostowskiRobinson53} this was first used, but not spelled out in detail.
Our approach follows
\cite[Section 2]{ar:MakowskyTARSKI}.
To keep it notationally simple we explain on an example.
Let $\tau$ be a vocabulary consisting of one binary relation symbol $R$,
$\sigma$ be a vocabulary consisting of one ternary relation symbol $S$.
In general, if $\tau$ and $\sigma$ are purely relational vocabularies
the definition can be extended in a straighforward way. If the contain function symbols (and constants)
one has to be a bit more careful when extending the definitions below.
However, for our purpose here, this is not needed.

We want to interpret a $\sigma$ structure on $k$-tuples of elements of a $\tau$-structure.

A {\em $\tau-\sigma$-translation scheme $\Phi =(\phi, \phi_S)$} consists of a $\tau$-formula $\phi(\bar{x})$
with $k$ free variables and a formula $\phi_S$ with $3k$ free variables.
$\Phi$ is quantifier-free if all its  translation formulas are quantifier-free.

Let $\cA =\langle A, R^A \rangle$ be a $\tau$-structure. 
We define a $\sigma$-structure $\Phi^*(\cA)=
\langle B, S^B \rangle$
as follows:
The universe is given by 
$$
B= \{ \bar{a} \in A^k : \cA \models \phi(\bar{a} \}
$$
and 
$$
S^B = \{ \bar{b} \in A^{k\times 3}: \cA \models \phi_S(\bar{b} \}
$$ 
$\Phi^*$ is called a {\em transduction}\footnote{
This terminology was put forward in the many papers by B. Courcelle, cf. \cite{bk:CourcelleEngelfriet2011}.
}.

Let $\theta$ be a $\sigma$-formula.
We define a $\tau$-formula $\Phi^{\sharp}(\theta)$ inductively by substituting
occurrences of $S(\bar{b})$ by their definition via $\phi_S$
where the free variables are suitable named.
$\Phi^{\sharp}$ is called a {\em translation}.

The fundamental property of translation schemes, transductions and
translation is the following:

\begin{proposition}[Fundamental Property of Translation Schemes]
\label{pr:Fund}
\label{pr:fundamental}
Let $\Phi$ be a $\tau-\sigma$-translation scheme, and $\theta$ be a $\sigma$-formula,
hence $\Phi^{\sharp}(\theta)$ is a $tau$-formula.
$$
\cA \models \Phi^{\sharp}(\theta)  \mbox{   iff   } \Phi^*(\cA) \models \theta
$$
If $\theta$ has free variables, the assignment have to be chosen accordingly.
Furthermore, if $\Phi$ is quantifier-free,
and
$\theta$ is a universal formula, $\Phi^{\sharp}(\theta)$ is also universal.
\end{proposition}

\ifskip
\else
\begin{proposition}
\begin{enumerate}[(i)]
\item
The definition of a plane $\Pi_{\cF}$ using coordinates in a field $\cF$
is given by a translation scheme $PP$.
\item
The definition of a field $\cF_{\Pi}$ using segment arithmetic
is given by a translation scheme $FF$.
\end{enumerate}
\end{proposition}
\fi 

In order to use translation schemes to prove decidability and undecidability of theories
we need two lemmas. 

\begin{lemma}
\label{le:onto}
Let $\Phi$ be a $\tau-\sigma$-translation scheme.
\begin{enumerate}[(i)]
\item
Let $\cA$ be a $\tau$-structure.
If the complete first order theory 
$T_0$ of $\cA$
is decidable,
so is the complete first order theory 
$T_1$ of $\Phi^*(\cA)$.
\item
There is a  $\tau$-structure $\cA$ such that
the complete first order theory $T_1$ of
$\Phi^*(\cA)$ is decidable,
but the complete first order theory $T_0$ of $\cA$ is undecidable.
\item
If however, $\Phi^{\sharp}$ is onto, i.e., for every $\phi \in \FOL(\tau)$
there is a formula $\theta \in \FOL(\sigma)$ with $\Phi^{\sharp}(\theta)$ logically equivalent to $\phi$,
then the converse of (i) also holds.
\item
Let $T \subseteq \FOL(\tau)$ be a decidable theory and $T' \subseteq \FOL(\sigma)$
and $\Phi^*$ be such that
$\Phi^*|_{Mod(T)}: Mod(T) \rightarrow  Mod(T')$ be onto.
Then $T'$ is decidable.
\end{enumerate}
\end{lemma}
\begin{proof}
(i): This follows from \ref{pr:Fund}.
$\cA \models \Phi^{\sharp}(\theta)$  iff $\Phi^*(\cA) \models \theta$,
hence, $\Phi^{\sharp}(\theta) \in T_1$ iff $\theta \in T_0$.
As $T_0$ is decidable, we can decide whether $\Phi^{\sharp}(\theta) \in T_0$,
and also, whether $\theta \in T_1$.
\\
(ii)
Let $\cA =\langle \N, +_N, \times_N \rangle$ where addition and multiplication are
ternary relations. $T_0(\cA)$ is undecidable by G\"odel's Theorem.

Now let $\Phi^*(\cA)$ be $\langle \N, +_A, \times_A \rangle$
where $+_A = +_N$ but $\times_A = +_N$.
$\Phi^*(\cA)$ is like Pressburger Arithmetic, but has two names ($+_A$ and $\times_A$)
for the same addition. Hence the complete theory of $\Phi^*(\cA)$ is decidable.
\\
(iii): If we assume $T_0$ to be decidable, we can only decide whether $\phi \in T_1$
for $\phi$ of the form $\phi = \Phi^{\sharp}(\theta)$.
\\
(iv): Let $\theta \in \FOL(\sigma)$. We want to check whether $T' \models \theta$.

Let $\cB \models T'$.

As $\Phi^*$ is onto, there is $\cA$ with  $\cA \models T$ and $\Phi^*(\cA) = \cB$.

Now we have, using Proposition \ref{pr:fundamental}
$$
\cB \models \neg\theta 
\mbox{  iff  }  
\cA \models \Phi^{\sharp}(\neg\theta)
\mbox{  iff  }  
T' \not \models \theta
\mbox{  iff  }  
T \not \models  \Phi^{\sharp}(\theta)
$$
But by assumption $T$ is decidable, hence $T'$ is decidable.
\hfill $\Box$
\end{proof}

\begin{remark}
The condition that $\Phi^{\sharp}$, resp. $\Phi^*$ have to be onto is often overlooked
in the literature\footnote{Theorems 1.36 and 1.37 as stated in \cite{balbiani2007logical} are only true
when one notices that their Theorems 1.20 and 1.21 imply that the particular transductions
used in Theorems 1.36 and 1.37 are indeed onto. However, this is not stated there.
}.
\end{remark}

We shall need one more observation:

\begin{lemma}
\label{le:deduction}
Let $T \subseteq \FOL(\tau)$ and $\phi \in \FOL(\tau)$.
Assume $T$ is decidable. Then $T \cup \{ \phi \}$ is also decidable.
\end{lemma}
\begin{proof}
This follows from the semantic version of the Deduction Theorem of First Order Logic:
$$
T  \cup \{ \phi \} \models \theta  \mbox{ iff  } T \models  (\phi \rightarrow \theta)
$$
\ \hfill $\Box$
\end{proof}

\subsection{Interpretability}
\ifskip\else
Recall that
a theory (set of formulas)a $T \subseteq \FOL(\tau)$
$T$ is {\em decidable}
if the set of consequences of $T$ 
is computable.
$T$ is  {\em axiomatizable} 
if the set of consequences of $T$ 
is computably enumerable.
$T$ is {\em complete}
if for every formula $\phi \in \FOL(\tau)$ without free variables either 
$T \models \phi$ or
$T \models \neg \phi$.
We note that if $T$ is axiomatizable and complete, the $T$ is decidable.
\fi 

A theory $T \subseteq \FOL(\tau)$
is  {\em  finitely axiomatizable} 
if there is a finite $T'$ which is axiomatizable
and has the same set of consequences as $T$.
$T$ is {\em essentially undecidable} 
if no theory $T' \subseteq \FOL(\tau)$ extending $T$
is decidable.
$T$ is {\em completely undecidable} 
if there is a finite subtheory $T' \subseteq T$ which is essentially undecidable.

Let $S \in \FOL(\sigma)$ and 
$T \in \FOL(\tau)$ be two theories over disjoint vocabularies. 
$S$ is {\em interpretable} in $T$,
if there exists a first order
translation scheme
$\Phi$ such that 
$$
\Phi^*(T) \models S.
$$
$S$ is {\em weakly interpretable} in $T$,
if there exists a theory $T'$ over the same vocabulary as $T$,
and a
translation scheme
$\Phi$ such that 
$$
\Phi^*(T') \models S.
$$

\begin{lemma}[{\cite[Statement (e3) on page 602]{bk:Beth}}]
\label{le:beklemishev}
Assume $S$ is a theory which is
\begin{enumerate}[(i)]
\item
finitely axiomatizable,
\item
essentially undecidable, and
\item
weakly interpretable in a theory $T$ using a translation scheme $\Phi$.
\end{enumerate}
Then $T$, and every subtheory of $T$, is undecidable.
\\
Moreover, there is a theory $T'$ with $T \subseteq T'$ and
with the same vocabulary as $T$, which is essentially undecidable.
\end{lemma}

Let $M$ be a class of $\tau$-structures closed under isomorphisms.
A 
$\tau-\sigma$-translation scheme $\Phi$ is {\em invertible on $M$}
if there  exists a
$\sigma-\tau$-translation scheme $\Psi$ such that for all $\cA \in M$
$$
\Psi^*(\Phi^*(\cA)) \simeq \cA
$$
and for all $\cB \in \Phi^*(M)$
$$
\Phi^*(\Psi^*(\cB)) \simeq \cB.
$$
Clearly, if $\Phi$ is invertible on $M$,
$\Phi^*|_M: M \rightarrow  \Phi^*(M)$ is onto.

\begin{lemma}
\label{le:interpretability}
Let $\cA$ be a $\sigma$-structure and 
$\cA'$ be a $\tau$-structure, and let $\Phi$ be a $\tau-\sigma$-translation scheme such that
$\Phi^*(\cA') = \cA$. Let $S$ be the complete theory of $\cA$. Assume $S$ is undecidable.
Let $T \subseteq \FOL(\tau)$ with $\cA' \models T$. 
and assume that $\Phi$ is invertible on $M = \{ \cA : \cA \models T \}$.
Then
\begin{enumerate}[(i)]
\item
$S$ is weakly interpretable in $T$, and 
\item
$T$ is undecidable.
\end{enumerate}
\end{lemma}
\begin{proof}
(ii) follows from (i) and Lemma \ref{le:beklemishev}.
\\
To see (i) we use that $\cA' \models T$ and use as $T'$ the complete theory of $\cA'$.
Now the invertibility of $\Phi^*$ allows us to complete the argument.
\hfill $\Box$
\end{proof}

In \cite{rautenberg1962unentscheidbarkeit} Lemma 
\ref{le:interpretability} is stated without the invertibility assumption
as the {\em Interpretationstheorem}. In the particular application in
\cite{rautenberg1962unentscheidbarkeit}, $S$ is the complete theory of the field of rational numbers,
which is undecidable by Proposition \ref{th:JRobinson}. The translation scheme $\Phi$ is vaguely sketched as $PP$,
and its inverse is not defined at all. We will show in the next section that both $PP$ and $RR$ are first order definable.
Theorem \ref{th:szmielew} implies that both $PP$ and $RR$
are invertible. This allows us to complete the gap 
in \cite{rautenberg1962unentscheidbarkeit}
in the proof of Theorem \ref{th:affine-undec}.
However, Theorem \ref{th:szmielew} only appears explicitly in \cite{blumenthal1980modern} and in
\cite{szmielew1983affine} and were not available in 1962.

\section{The role of coordinates}
\label{se:coord}
\subsection{Analytic geometry over fields of characteritic $0$}
\label{se:analytic}
Given a field $\cF$ or an ordered field $\cO$
we define the following relations in $\cF$ ($\cO$), where elements 
$P=(x,y)$ are called {\em points}
and $ \ell = (a,b,c) =\{ (x,y) : ax+by=c \}$ are called {\em lines}.
Similarly, we write $P_i =(x_i,y_i)$
and $\ell_i =(a_i,b_i,c_i)$.
In $\tau_{field}$ points are defined using a quantifier-free formula and lines are defined 
using an existential formula.
In $\tau_{f-field}$ both are defined using a quantifier-free formula.

\begin{description}
\item[\bf Incidence:]
$P \in \ell$ iff $ax+by+c=0$.
In $\tau_{f-field}$ is a quantifier-free formula.
\item[\bf Equidistance:]
$Eq(P_1,P_2, P_3, P_4)$ iff
$ 
(x_1 -x_2)^2 + (y_1 -y_2)^2 =
(x_3 -x_4)^2 + (y_3 -y_4)^2 
$.
In $\tau_{f-field}$ is a quantifier-free formula.
\item[\bf Orthogonality:]
$Or(\ell_1, \ell_2)$ (or $\ell_1 \perp \ell_2$) iff $a_1a_2 +b_1b_2=0$.
In $\tau_{f-field}$ is a quantifier-free formula.
\end{description}
For equiangularity we have to work a bit more.
Let $\ell =(a,b,c)$ be a line. The {\em slope of $\ell$} is defined as $sl(\ell)= \frac{a}{b}$.
Now let $\ell_1, \ell_2$ be two lines intersection at the point $p$, 
let and $k$ a line with $Or(k, \ell_1)$
intersecting $\ell_i$ at $Q_i$ ($=1,2)$.
The angle $\angle(Q_1, P, Q_2)$ is an acute angle.
For acute angles
we define 
$$
\tan(Q_1,P,Q_2) = \left| \frac{sl(\ell_1) - sl(\ell_2)}{1+ sl(\ell_1)sl(\ell_2)} \right|.
$$
We now give a quantifier-free definition of equiangularity in rectangular triangles.

\begin{description}
\item[\bf Rectangular:]
$rectangular(P_1, P_2, P_3)$  iff $Or((P_1, P_2), (P_1, P_3))$.
\item[\bf Equiangular:]
Assume we have two rectangular triangles
$P_1P_2P_3$ and
$Q_1Q_2Q_3$
with
$rectangular(P_1, P_2, P_3)$  
and
$rectangular(Q_1, Q_2, Q_3)$  
we define
$An(P_1,P_2,P_3,Q_1,Q_2,Q_3)$ iff  $\tan(P_1,P_2,P_3) =  \tan(Q_1,Q_2,Q_3)$.
\\
In $\tau_{f-field}$ this is a quantifier-free formula.
\end{description}

\ifskip
\else
The following are definable using $Eq$ and  $Or$ with existential formulas over $\tau_{f-field}$:
\begin{description}
\item[\bf L-equidistant:] $Leq(P_1, \ell, P_2)$ iff there exists points $Q_1, Q_2 \in \ell$ such that
$Or((P_1,Q_1), \ell)$,
$Or((P_2,Q_2), \ell)$ and 
$Eq(P_1,Q_1, P_2, Q_2)$. 
\item[\bf Symmetric Line:]
$SymLine(P_1, \ell, P_2)$ iff there is a point $Q \in \ell$ such that
$Or((P_1,q), \ell)$,
$Or((P_2,q), \ell)$ and 
$Eq(P_1,q, P_2, q)$. 
\end{description}
\fi 

If the field is an ordered field we define additionally:
\begin{description}
\item[\bf Colinear:]
$Col(P_1, P_2, P_3)$ iff $\exists \ell (\bigwedge_{i=1}^3 P_i \in \ell)$.
\\
For $\ell =(a,b,c)$ and $P_i =(x_i, y_i)$ we can write this as
$$
\exists a, b, c (\bigwedge_{i=1}^3 ax_i + bxy_i +c = 0 )
$$
which is equivalent to
$$
\det
\left(
\begin{array}{ccc}
x_1 & y_1 & 1 \\
x_2 & y_2 & 1 \\
x_3 & y_3 & 1 
\end{array}
\right)
= 0
$$
which in $\tau_{f-field}$ is a quantifier-free formula.
\item[\bf Betweenness:]
$
Be(P_1, P_2, P_3)$ iff 
\begin{gather}
Col(P_1, P_2, P_3)
 \wedge 
\notag \\
\left[ 
\left(
(x_1 \leq x_2 \leq x_3) \wedge (y_1 \leq y_2 \leq y_3)
\right)
\vee
\left(
(x_3 \leq x_2 \leq x_1) \wedge (y_3 \leq y_2 \leq y_1)
\right)
\right]
\notag
\end{gather}
\\
which in $\tau_{f-ofield}$ is a quantifier-free formula.
\end{description}

\begin{definition}
\label{transduction-PP}
Given a field with universe $A$, let $Points(\cF) = A^2$, and $Lines(\cF) = A^3$. 
For a field $\cF$, respectively an ordered field $\cO$, we define
\begin{enumerate}[(i)]
\item
$\Pi_{\in}(\cF)$ to be the two sorted structure
$$
\langle Points(\cF), Lines(\cF); \in_{\cF} \rangle.
$$
The quantifier-free first order translation scheme $PP_{\in} = (Lines, \in)$ 
satisfies
$PP_{\in}^*(\cF) = \Pi_{\in}(\cF)$. 
\item
$\Pi_{wu}(\cF)$ to be the two sorted structure
$$
\langle Points(\cF), Lines(\cF); \in_{\cF}, Eq_{\cF}, Or_{\cF} \rangle.
$$
The quantifier-free first order translation scheme $PP_{wu} = (Lines, \in, Eq, Or)$ 
satisfies
$PP_{wu}^*(\cF)= \Pi_{wu}(\cF)$.
\item
$\Pi_{hilbert}(\cO)$ to be the two sorted structure
$$
\langle Points(\cO), Lines(\cO); \in_{\cO}, Eq_{\cO}, An(\cO), Be_{\cO} \rangle.
$$
The quantifier-free first order translation scheme $PP_{hilbert} = (Lines, \in, Eq, An, Be)$ 
satisfies
$PP_{hilbert}^*(\cF)= \Pi_{hilbert}(\cF)$.
\end{enumerate}
\end{definition}

This gives us:
\begin{proposition}
\label{pr:PP}
The translation schemes 
$PP_{\in}$,
$PP_{wu}$ and
$PP_{hilbert}$
are quantifier-free first order translation schemes.
\end{proposition}
\ifskip
\else
\begin{proof}
All the relations $\in, Eq, An, Be, Or$ are quantifier-free first order definable in the 
vocabulary of (ordered) fields.
\hfill $\Box$
\end{proof}
\fi 

\subsection{Properties of $PP_{\in}$ and $PP_{wu}$}

We summarize now the properties needed of these translation schemes and their induced
transductions and translations. Here, and in the next section we call these properties
the {\em correctness} of the translation schemes, because they state that they behave as needed
to prove undecidability results.

\begin{theorem}[Correctness of $PP_{\in}$ and $PP_{wu}$]
\label{th:cart-a}
\ 
\begin{enumerate}[(i)]
\item
(\cite[14.1]{bk:Hartshorne2000})
If $\cF$ is a field, then $PP_{\in}^*(\cF)$ satisfies the incidence axioms $(I_1)- (I_3)$,
the Parallel Axiom and the Pappus Axiom.
\item
(\cite[14.4]{bk:Hartshorne2000})
If $\cF$ is a field of characteristic $0$, then $PP_{\in}^*(\cF)$ satisfies additionally the Axioms of Infinity,
i.e., is an infinite Pappian plane.
\item
(\cite{bk:Wu1994})
If $\cF$ is a Pythagorean field of characteristic $0$, then $PP_{Wu}^*(\cF)$ satisfies
(I-1)-(I-3), 
(O-1) -(O-5),
the Parallel Axiom, the Axiom of Infinity, the Axiom of Desargues and the Axiom of Symmetric Axis,
which are axioms of a {\em metric Wu plane}.
\item
(\cite{alperin2000mathematical})
If $\cF$ is a Vieta field, then $PP_{Wu}^*(\cF)$ satisfies the Huzita-Hatori axioms (H-1)-(H-7).
\end{enumerate}
\end{theorem}

\begin{theorem}[Correctness of $PP_{hilbert}$]
\label{th:cart-o}
\ 
\begin{enumerate}[(i)]
\item
(\cite[17.3]{bk:Hartshorne2000})
If $\cO$ is an ordered Pythagorean field, then $PP_{hilbert}^*(\cO)$ 
satisfies
(I-1) - (I-3), 
(B-1) - (B-4)
(C-1) -(C-6)
and the Parallel Axiom,
which are axioms of
a {\em Hilbert Plane 
which satisfies the parallel axiom}.
\item
(\cite[17.3]{bk:Hartshorne2000})
If $\cO$ is an ordered Euclidean field, then $PP_{hilbert}^*(\cF)$ is a Hilbert Plane 
which satisfies the parallel axiom and Axiom E.
\end{enumerate}
\end{theorem}

\subsection{Introducing coordinates}
\label{se:coordinates}

We have seen in the last section how get models of geometry using coordinates in a field.
Now we want to find a way to define coordinates from a model $\Pi$ of geometry.
We say that we want to {\em coordinatize $\Pi$}.
This problem has a long tradition and was solved already in the 19th century.

There are two accepted ways of coordinatizing: If we have the notion of equidistance and betweenness available,
we can define an arithmetic of line segments.
This is discussed in detail in \cite[Chapter 18]{bk:Hartshorne2000}.
In the absence of betweenness and congruence,
but in the presence of the Parallel Axiom, one can use Pappus' Axiom
to define the arithmetic operations even in a Pappus plane.
This was first done
by K.G.C. von Staudt \cite{von1847geometrie,von1857beitrage}, a student of C.F. Gauss, 
before D. Hilbert's \cite{hilbert1902foundations}.
The first modern treatment of coordinatization for affine and projective planes was given
by M. Hall \cite{hall1943projective}.

\begin{definition}
Let $\tau$ a vocabulary for geometry, 
$T \subseteq \FOL(\tau)$  a set of axioms of geometry, 
$T_f$ be a set of axioms for fields in 
$\tau_{fields}$
or $\tau_{ofields}$.
We say that the models of $T$ have a {\em first order coordinatization}
in fields satisfying  $T_f$
if
there exists a first order translation scheme $CC_{field}$ such that
\begin{enumerate}[(a)]
\item
for every $\Pi$ which satisfies $T$ the structure 
$CC_{field}^*(\Pi)$ 
($CC_{o-field}^*(\Pi)$) 
is a field
which satisfies $T$; 
\item
for every field $\cF$ which satisfies $T_f$, the $\tau$-structure $PP_{\tau}(\cF)$
satisfies $T$;
\item
For every field $\cF$  which satisfies $T_f$ we have
$$
CC_{field}(PP_{\tau}(\cF) \simeq \cF;
$$
\item
For every $\tau$-structure $\Pi$ which satisfies $T$ we have
$$
PP_{\tau}(CC_{field}(\Pi)) \simeq \Pi.
$$
\end{enumerate}
\end{definition}

We have formulated the definition in terms of the relational vocabularies for fields
to make the use of translation schemes simple. As we deal here with full first order logic,
there is no loss of generality.

In order to deduce undecidability of geometric theories using undecidability
of theories of fields we will need the following:

\begin{theorem}[Segment Arithmetics]
\label{th:segment}
Every P-Hilbert plane has a first order coordinatization  $FF_{field}$ (via segment arithmetic).
\end{theorem}

\begin{theorem}[Planar Ternary Rings]
\label{th:PTR}
Every infinite Pappus plane without finite lines has a first order coordinatization $RR_{field}$ (via planar ternary rings).
\end{theorem}

We will show in the sequel that $FF_{field}$ and 
$RR_{field}$ are $\FOL$-definable.

\subsection{Segment arithmetic}

Given a Hilbert plane $\Pi$ which satisfies the Parallel Axiom,
we now want to show that one can interpret in $\Pi$ an ordered field of coordinates
$\cF_{hilbert}(\Pi)$.
Note that orthogonality $Or(\ell_1, \ell_2)$ of lines is definable
in every Hilbert plane using equiangularity.
We follow essentially \cite[Chapter 4]{bk:Hartshorne2000}.

Fix a line segment $1=[A_0,A_1]$ given by two points $A_0, A_1$.

We first define commutative semiring $\cS_{hilbert}(\Pi)$
as follows:
\begin{description}
\item[Positive elements:]
Equivalence classes $[P_1,P_2]$ of pairs of points $P_1,P_2$ with $Eq(P_1,P_2)$.
\item[Zero element:]
The equivalence class $[A_0, A_0]$.
\item[Unit element:]
The equivalence class $[A_0, A_1]$.
\item[Positive addition:] 
Choose three points $P_1, P_2, P_3$ such that $Be(P_1, P_2, P_3)$.
Then we put $[P_1,P_2] + [P_2, P_3] = [P_1, P_3]$.
If $P_1, P_2, P_3$ or not colinear, we always can choose
$P_1', P_2', P_3'$ with $Be(P_1', P_2', P_3')$
such that
$[P_1,P_2] = [P_1', P_2']$ and
$[P_2,P_3] = [P_2', P_3']$. 
\item[Positive multiplication:] 
Let $P_0, P_1, P_2, P_3, P_4$ be points such that
$Be(P_0, P_1, P_2)$ and $Be(P_0, P_3, P_4)$ and
the lines $(P_1, P_3)$ and $(P_2, P_4)$ are parallel
and the lines $(P_1, P_2)$ and $(P_3, P_4)$ are orthogonal,
and $[P_0, P_3] = [A_0,A_1] =1$ is the unit length.
Then we put for 
$a=[P_0, P_1]$ and
$b=[P_0, P_4]$ the product $ab=[P_0, P_2]$.
\end{description}

One easily verifies now:
\begin{proposition}
\label{le:fol-segmentarithmetic}
The arithmetic operations defined as above are definable in $\FOL$ in the vocabulary $\tau_{hilbert}$.
\end{proposition}

\begin{proposition}
\label{le:addition}
In any Hilbert plane (even without the Parallel Axiom) addition of line segments
as defined above is well-defined, commutative, associative.
Furthermore, if $a,b$ are two line segments, then one of the following holds:
\begin{enumerate}[(i)]
\item $a=b$,
\item
There is a line segment $c$ such that $a+c=b$,
\item
There is a line segment $d$ such that $a=b+d$,
\end{enumerate}
\end{proposition}

\begin{proposition}
\label{le:multiplication}
In any Hilbert plane $\Pi$ with the Parallel Axiom multiplication of line segments
as defined above is well-defined, associative, and has $1$ as its neutral element.
Furthermore, for all line segments $a,b,c$ we have
\begin{enumerate}[(i)]
\item
$a(b+c)= ab +ac$
\item
There is a unique $d$, such that $ad=1$,
\item
If $\Pi$ is also Pappian, then multiplication is also commutative.
\end{enumerate}
\end{proposition}

\begin{lemma}
In every semiring satisfying Propositions 
\ref{le:addition}
and
\ref{le:multiplication}
we can define an ordered field.
In fact this field is definable using a first order translation scheme.
\end{lemma}
\begin{proof}
We use the standard construction the same way as one constructs the ordered field of
rational numbers $\Q$ from the ordered semiring of the natural numbers $\N$.
\hfill $\Box$
\end{proof}

This gives us the first order translation schemes 
$FF_{field}$ and $FF_{ofield}$. 


\begin{theorem}[Correctness of $FF_{field}$ and $FF_{ofield}$]
\label{th:corr-segment}
\ \\
Let $\Pi$ be a Hilbert Plane  which satisfies the Parallel Axiom.
\begin{enumerate}[(i)]
\item
$FF_{field}^*(\Pi)$ is a field of characteristic $0$ which
can be uniquely ordered to be an ordered field $\cF_{\Pi}$.
\item
Let $\cF_{\Pi}=FF_{field}^*(\Pi)$ be the ordered field of segment arithmetic in $\Pi$.
Then $\cF$ is Pythagorean and $PP_{hilbert}^*(\cF)$ is isomorphic to $\Pi$.
\item
An ordered field $\cO$ 
is Pythagorean  iff
$PP_{hilbert}^*(\cO)$  is a Hilbert Plane which satisfies
the Parallel Axiom.
\item 
$\Pi$ is a Euclidean plane iff $FF_{field}^*(\Pi)$ is a Euclidean field.
\item
$\cF$ is a Euclidean fields iff $PP_{hilbert}^*(\cF)$ is a Euclidean plane.
\end{enumerate}
\end{theorem}
\begin{proof}
A proof may be found in
{\cite[Theorems 20.7, 21.1 and 21.2]{bk:Hartshorne2000}}.
\hfill
$\Box$
\end{proof}

\subsection{Planar ternary rings}
\label{se:ptr}
In order to use the undecidability of the theory of fields,
we have to find a first order transduction $RR_{\in}$ which turns any Pappian plane 
$\Pi$ into a field $RR_{\in}(\Pi)$ without using the betweenness
relation $Be$.
Fortunately, this can be done using Planar Ternary Rings, which were introduced by M. Hall 
in \cite{hall1943projective}. 
M. Hall credits \cite{von1857beitrage,hilbert1971foundations} for the original idea.
A good exposition can be found in \cite{blumenthal1980modern,szmielew1983affine}.
We follow here almost verbatim \cite{ivanov2016affine}, 
which is a particularly nice exposition of \cite{hall1943projective}.

Let $\Pi$ be an affine plane satisfying (I-1)- (I-3) (ParAx),
with two distinguished lines $\ell_0, m_0$ in $\Pi$. 
Let $O$ be the point of intersection of $\ell_0$ and $m_0$.

\begin{lemma}
\label{bijection}
There is a formula $bij(x,y,d) \in \FOL_{\in}$ which
for every line $\delta$
different from $\ell_0$ and $m_0$
such that $O \in \delta$ 
defines a bijection between $\ell_0$ and $m_0$.
\end{lemma}
\begin{proof}
Let $x \in \ell_0$ and $z(x)$ be the intersection with $\delta$ of the line $m_1$ parallel to $m_0$ containing $x$.
Let $y(x) \in m_0$ be the intersection of the line $\ell_1$ parallel to $\ell_0$ containing $z(x)$.
Clearly $f_{\delta}: \ell_0 \rightarrow m_0$ given by $f_{\delta}(x) =y(x)$ is a bijection and is $\FOL$ definable
by a formula $bij(x,y,\delta)$.
\hfill $\Box$
\end{proof}

We will define a structure $RR_{\Pi}$ with universe  a set $K$ (which we take to be $\ell_0$).
Thinking of $\ell_0$ and $m_0$ as axes of a coordinate system we can identify the points of $\Pi$
with pairs of points in $K^2$. The projection of a point $P$ onto $ \ell_0$ is defined by the point $x \in \ell_0$
which is the intersection of the line $m_1$ parallel to $m_0$ with $P \in m_1$.
The projection of a point $P$ onto $m_0$ is defined analogously.
The point $0$ has coordinates $(0,0)$.
Furthermore, we fix an arbitrary point $1 \in \ell_0$ different from $0$ which has coordinates $(1,0)$.

Next we define the {\em slope} of a line $\ell$ in $\Pi$ to be an element $sl(\ell) \in K \cup \{\infty\}$
If $\ell$ is parallel to $\ell_0$ its slope is $0$ and it is called a {\em horizontal} line.
If $\ell$ is parallel to $m_0$ its slope is $\infty$ and it is called a {\em vertical} line.
For $\ell$ not vertical, let $\ell_1$ be the line parallel to $\ell$ and passing through $0$.
Let $(1,a)$ be the coordinates of the intersection of $\ell_1$ with the line vertical line $\ell_2$
passing through $(1,0)$.
Then the slope $sl(\ell)= a$.

This shows:
\begin{lemma}
\label{slope}
There is a first order formula $slope(\ell,a,\delta) \in \FOL_{\in}$ which expresses $sl(\ell)=a$.
with respect to the auxiliary line $\delta$.
\end{lemma}

\begin{lemma}
\label{slope-1}
\begin{enumerate}[(i)]
\item
Two lines $\ell, \ell_1$ have the same slope, $sl(\ell)=sl(\ell_1)$ iff they are parallel.
\item
For the line $\delta$ we have $sl(\delta)=1$ (because $(1,1) \in \delta$).
\end{enumerate}
\end{lemma}

We now define a ternary operation $T: K \rightarrow K$.
We think of $T(a,x,b)= \langle ax+y \rangle$ as the result of multiplying $a$ with $x$ and then adding $b$.
But we yet have to define multiplication and addition.

Let $a,b,x \in K$. Let $\ell$ be the unique line with $sl(\ell)=a \neq \infty$ 
intersecting the line $m_0$ at the point $P_1 =(0,b)$.
Let $\ell_1 =\{ (x,z) \in K^2 : z \in K\}$. For every $x \in K$ the line $\ell$ intersects $\ell_1$
at a unique point, say $P_2= (x,y)$.
We set $T(a,x,b) = y$.

\begin{lemma}
\label{ptr}
There is a formula $Ter(a,x,b,y,\delta) \in \FOL_{\in}$
which expresses that $(a,x,b)=y$ with respect to the auxiliary line $\delta$.
\end{lemma}

\begin{lemma}
\label{ptr-1}
The ternary operation $T(a,x,b)$ has the following properties and interpretations:
\begin{description}
\item[T-1:]
$T(1,x,0)=T(x,1,0)=x$
\\
$T(1,x,0)=x$ means that the auxiliary line $d =\{ (x,x) \in K^2 : x \in K \}$
is a line with $sl(d)=1$.
\\
$T(x,1,0)=x$ means that the slope 
of the line $\ell$ passing through $(0,0)$ and $(1,x)$ is given by $sl(\ell)=x$.
This is the true interpretation of the slope in analytic geometry.
\item[T-2:]
$T(a,0,b)=T(0,a,b)=b$
\\
The equation $T(a,0,b)=b$ means that the line $\ell$ defined by $T(a,x,b)=y$
intersects $m_0$ at $(0,b)$ (which is the meaning of $ax+b$ in analytic geometry).
\\
The equation $T(0,a,b)=b$ means that the horizontal line $\ell_1$ passing through $(0,b)$
consists of the points $\{ (a,b) \in K^2 : a \in K \}$.
\item[T-3:]
For all $a,x,y \in K$ there is a unique $b \in K$ such that $T(a,x,b) =y$
\\
This means that for every slope  $s$ different from $\infty$ there is a unique line
$\ell$ with $sl(\ell)=s$ passing through through $(x,y)$.
\item[T-4:]
For every $a, a', b, b' \in K$ and $a \neq a'$ the equation $T(a,x,b) = T(a'x,b')$ has a unique
solution $x \in K$.
\\
This means that two lines $\ell_1$ and $\ell_2$ with different slopes not equal to $\infty$ 
intersect at a unique point $P$.
\item[T-5:]
For every $x, y, x', y' \in K$ and $x \neq x'$ there is a unique pair $a,b \in K$ such that
$T(a,x,b)=y$ and 
$T(a,x',b)=y'$. 
\\
This means that any two points $P_1, P_2$ not on the same vertical line are contained 
in a unique line $\ell$ with
slope different from $\infty$.
\end{description}
\end{lemma}

A structure $\langle K, T_K \rangle$  with a ternary operation $T_K$ satisfying (T-1)-(T-5)
is called a {\em planar ternary ring} PTR.
If the PTR
$\langle K, T_K \rangle$ 
 arises from a Pappian plane
we define addition by $add_T(a,b,c)$ by $T(a, 1, b) =c$
and multiplication by $mult_T(a,x,c)$ by $T(a, x, 0) =c$.

We define now the translation schemes
$RR_{ptr} = (line, Ter)$ and
$RF_{field} = (line, add_T, mult_T)$.
The transduction $RR_{ptr}^*$ maps  incidence planes into structures with universe defined by the lines
and a ternary function,  and the transduction
$RF_{ptr}^*$ maps Pappian planes into structures with universe defined by the lines and with two binary operations.

With these definitions we get:
\begin{theorem}[Correctness of $RR_{ptr}$and $RF_{field}$]
\label{pr:ptr}
Let $\Pi$ be plane satisfying I-1, I-2, and I-3 
with distinguished lines $\ell, m, d$ and points $O=(0,0)$ and $I=(1,0)$.
\begin{enumerate}[(i)]
\item
$RR_{ptr}^*(\Pi)$ is a planar ternary ring.
\item
$\Pi$ is a (infinite) Pappian plane iff  
$RF_{field}^*(\Pi)$ is a field (of characteristic $0$).
\end{enumerate}
\end{theorem}
A detailed proof may be found in \cite{blumenthal1980modern,szmielew1983affine}.


\section{Undecidable geometries}
\label{se:undecidable}
\subsection{Incidence geometries}

First we look $\tau_{\in}$-structures, i.e., at models of the incidence relation alone.
To prove undecidability, the correctness of the translation scheme $RR_{field}$,  Theorem \ref{pr:ptr},
is not enough. We still have to show that
$RR_{field}^*$ is onto as a transduction from Pappus planes to fields.

\begin{theorem}[{\cite[Section 4.5]{szmielew1983affine}}]
\label{th:szmielew}
\ 
\begin{enumerate}[(i)]
\item
If $\Pi$ is a Pappus plane
there is a field $\cF_{\Pi}$ such that $PP_{\in}^*(\cF_{\Pi})$ is isomorphic to $\Pi$.
\item
If additionally $\Pi$ satisfies (Inf),
$\cF_{\Pi}$ is a field of characteristic $0$.
\end{enumerate}
In fact, $\cF_{\Pi}$ can be chosen to be $RR_{field}^*(\Pi)$ from Theorem \ref{th:PTR}.
\end{theorem}

\begin{corollary}
\label{cor:affine}
$RR_{field}^*$ is onto as a transduction from Pappus planes to fields.
\end{corollary}

We now can use Proposition \ref{th:JRobinson}(i) to prove Theorem \ref{th:affine-undec},
which states that the theory $T_{pappus}$
of Pappus planes is undecidable.

\begin{proof}[Proof of Theorem \ref{th:affine-undec}]
Assume, for contradiction, that $T_{pappus}$ is decidable.
By Corollary \ref{cor:affine}, $RR_{field}^*$ is onto, hence the theory of fields
is decidable, which contradicts Proposition \ref{th:JRobinson}.
\hfill $\Box$
\end{proof}

\begin{proof}[Alternative proof of Theorem \ref{th:affine-undec}]
We could also use Proposition \ref{th:JRobinson}(iii) to prove Theorem \ref{th:affine-undec}.
Let $\bQ$ be the field of rational numbers.
In this case we use Lemma \ref{le:interpretability}
with $\cA = PP^*(\bQ)$ and $\cA' = RR^*(\cA)$. $RR^*(\cA) \simeq \bQ$, by Theorem \ref{th:szmielew}.
So $S$ is the complete theory of $\bQ$, which is undecidable by Proposition \ref{th:JRobinson}(iii),
hence $T_{pappus}$ is undecidable.
\hfill $\Box$
\end{proof}

This alternative proof does not work for Hilbert planes and Euclidean planes, because we do not know whether
the complete theories of the fields $\bP$ and $\bE$ are undecidable.

\subsection{Hilbert planes and Euclidean planes}

\ifskip\else
\begin{theorem}[{\cite[Theorems 17.3]{bk:Hartshorne2000}}]
\ \\
Let $\cF$ be any ordered Pythagorean field.
\begin{enumerate}[(i)]
\item
Then $PP_{hilbert}^*(\cF)$ is a Hilbert plane which satisfies the Parallel Axiom.
\item
$PP_{hilbert}^*(\cF)$ is Euclidean iff $\cF$ is Euclidean.
\end{enumerate}
\end{theorem}

\begin{theorem}[{\cite[Theorems 20.7 and 21.1]{bk:Hartshorne2000}}]
\ \\
Let $\Pi$ be a Hilbert Plane  which satisfies the Parallel Axiom.
\begin{enumerate}[(i)]
\item
$FF_{field}^*(\Pi)$ is a field of characteristic $0$ which
can be uniquely ordered to be an ordered field $\cF_{\Pi}$.
\item
Let $\cF_{\Pi}=FF_{field}^*(\Pi)$ be the ordered field of segment arithmetic in $\Pi$.
Then $\cF$ is Pythagorean and $PP_{hilbert}^*(\cF)$ is isomorphic to $\Pi$.
\item
An ordered field $\cF$ 
is Pythagorean  iff
$PP_{hilbert}^*(\cF)$  is a Hilbert Plane which satisfies
the Parallel Axiom.
\end{enumerate}
\end{theorem}

\begin{theorem}[{\cite[Theorems 21.2]{bk:Hartshorne2000}}]
\ 
\begin{enumerate}[(i)]
\item 
$\Pi$ is a Euclidean plane iff $FF_{field}^*(\Pi)$ is a Euclidean field.
\item
$\cF$ is a Euclidean fields iff $PP_{hilbert}^*(\cF)$ is a Euclidean plane.
\end{enumerate}
\end{theorem}
\fi 

Similarly, the correctness of the translation scheme
$FF_{field}^*$, Theorem \ref{th:corr-segment}, is not enough. 

We need one more lemma\footnote{
In \cite{balbiani2007logical} it is overlooked that Lemma \ref{le:new-1} is
needed in order to apply Lemma \ref{le:onto}.
In \cite{bk:Schwabhaeuser1983} it is used properly but not explicitly stated.
}:

\begin{lemma}
\label{le:new-1}
Let $\cF$ be a Pythagorean field
and $\Pi_{\cF}= PP_{hilbert}^*(\cF)$.
Then $FF_{field}^*(\Pi_{\cF})$ is isomorphic to $\cF$.
\end{lemma}
\begin{proof}
For every $a \in \cF$ with $a \geq 0$ there is a segment in $\Pi_{\cF}$
of the form $[(0,0), (a,0)]$ on the line $y=0$,
and every segment $[(0,0), (b,0)]$
on the line $y=0$ corresponds to an element $b \in \cF$.
We conclude that there is an isomorphism $f$  of ordered fields between $\cF$ and the 
segments on the line $y=0$.
Similarly, there is an isomorphism  of ordered fields $g$ between the
segments on the line $y=0$ and the field $FF_{field}^*(\Pi_{\cF})$.
Composing the two isomorphisms gives the required isomorphism between
$\cF$ and $FF_{field}^*(\Pi_{\cF})$.
\hfill $\Box$
\end{proof}

\begin{corollary}
\label{co:onto-1}
\ 
\begin{enumerate}[(i)]
\item 
$FF_{field}^*$ is onto as a transduction from Hilbert planes to ordered Pythagorean fields.
\item 
$FF_{field}^*$ is onto as a transduction from Euclidean planes to ordered Euclidean fields.
\end{enumerate}
\end{corollary}

Using Ziegler's Theorem \ref{th:Ziegler}, 
Lemma \ref{le:onto},
Theorem \ref{th:corr-segment} and
and 
Corollary \ref{co:onto-1}
we conclude:
\begin{theorem}
\label{th:euclid-undec}
\ 
\begin{enumerate}[(i)]
\item 
The consequence problem for the Hilbert plane which satisfies the Parallel Axiom is undecidable.
\item 
The consequence problem for the Hilbert plane is undecidable.
\item 
The consequence problem for the Euclidean Plane is undecidable.
\end{enumerate}
\end{theorem}

\subsection{Wu's geometry}
There are two systems of orthogonal geometry, Wu's orthogonal geometry $T_{o-wu}$ and 
Wu's metric geometry $T_{m-wu}$.

Recall that in Wu's orthogonal geometry we have as basic relation
incidence $A \in l$ and Orthogonality $Or(l_1, l_2)$,
but neither betweenness nor equidistance.
We also require that models of Wu's geometry are infinite Pappian planes without finite lines.

\begin{theorem}
\label{th:orthogonal}
\begin{enumerate}[(i)]
\item
Let  $\cF$ be a Pythagorean field of characteristic $0$. 
Then $PP_{wu}(\cF)$ is a metric Wu plane. 
\item
Conversely, let $\Pi$ is a metric Wu plane
then $RF_{field}^*(\Pi)$
is a Pythagorean field of characteristic $0$.
\item
Furthermore, $RF_{field}^*(PP_{wu}(\cF))$ is isomorphic to $\cF$
and $PP_{wu}( RF_{field}(\Pi)$ is isomorphic to $\Pi$.
\end{enumerate}
\end{theorem}

\begin{corollary}
\label{cor:o-onto}
$RF_{field}^*$ maps metric Wu planes onto Pythagorean fields.
\end{corollary}

\begin{theorem}
\label{th:wu-undecidable}
\item
The consequence problem for (Wu-metric) is undecidable.
\item
The consequence problem for
(Wu-orthogonal) is undecidable.
\end{theorem}
\begin{proof}
(i):
Use Ziegler's Theorem \ref{th:Ziegler} for $ACF_0$,
Lemma \ref{le:onto},
Theorem \ref{pr:ptr}
and Corollary \ref{cor:o-onto}.
\\
(ii):
We observe that (Wu-metric) is obtained from (Wu-orthogonal)
by adding one more axiom. This gives that if (Wu-orthogonal) were decidable, so would (Wu-metric)
be decidable, which contradicts (i).
\hfill $\Box$
\end{proof}

\subsection{Origami geometry}

In Origami Geometry we have also points and lines,
the incidence relation $A \in l$, the orthogonality relation $Or(l_1, l_2)$,
 a relation $SymP(A,l,B)$ and a relation $d(A,l_1, l_2)$.

The intended interpretation of $SymP(A,l,B)$ states that
$A$ and $B$ are symmetric with respect to $l$, i.e., $l$ is perpendicular to the
line $AB$ and intersects $AB$ at a point $C$ such that $Eq(A,C)$ and $Eq(C,B)$.

The intended interpretation of $d(A,l_1, l_2)$ states that the point $A$ has the same 
perpendicular distance from $l_1$ and $l_2$.

Clearly, $SymP(A,l,B)$ and $d(A,l_1, l_2)$
are definable in a Hilbert plane using $\in, Eq, Or$.

An {\em Origami Plane} is an infinite Pappian plane without finite lines 
which satisfies additionally the axioms
(H-1) to (H-7).

\begin{theorem}[{\cite{alperin2000mathematical}}]
\label{th:origami}
If $\Pi$ is an Origami plane, then $RF_{field}^*(\Pi)$ is a Vieta field.
\\
Conversely, for every Vieta field $\cF$ the structure $PP_{origami}^*(\cF)$
is an Origami field.
\end{theorem}

\begin{corollary}
\label{cor:origami-onto}
$RR_{field}^*$ maps Origami planes onto Vieta fields.
\end{corollary}

Using Ziegler's Theorem \ref{th:Ziegler}, Lemma \ref{le:onto} and \ref{le:deduction}
we can now apply Theorem \ref{th:origami} and Corollary \ref{cor:origami-onto}
to conclude:

\begin{theorem}
\label{th:undec-origami}
\begin{enumerate}[(i)]
\item
The consequence problem for Origami Planes is undecidable.
\item
The consequence problem for the Huzita axioms (H-1)-(H-7) is undecidable.
\end{enumerate}
\end{theorem}

\section{Decidability for fragments of first order logic}
\label{se:decidable}
Problems in high-school geometry are usually of the form
\begin{quote}
Given a configuration between points $\bar{p}$ (and lines) described by a quantifier-free formula
$\phi(\bar{p})$ show that these points also satisfy a quantifier-free formula $\psi(\bar{p})$
$$
\sigma(\bar{p}):
\forall \bar{p} (\phi(\bar{p} \rightarrow \psi(\bar{p}))
$$
\end{quote}
A typical example would be:
\begin{quote}
Of the three altitudes $\ell_1,\ell_2, \ell_3$ of a triangle $P_1P_2P_2$ 
which intersect pairwise at the points 
$P_{1,2}$
$P_{1,3}$
$P_{2,3}$, show that 
$P_{1,1}=P_{1,2}=P_{1,3}$.
\end{quote}
The formula $\sigma$ is a universal Horn formula in $\fU\fH= \fU \cap \fH$.

In the literature the following was observed:

\begin{proposition}
\label{pr:universal}
\ \\
\begin{description}
\item[(\cite{kapur1988refutational})]
The universal consequences of $T_{m-wu}$ are decidable.
\item[(\cite{pambuccian1994ternary})]
The universal consequences of 
$T_{p-hilbert}$  and $T_{euclid}$ are decidable.
\end{description}
\end{proposition}

Using a simple model theoretic argument we can generalize this 
to Theorem \ref{th:universal} below.

\ifskip
\else
We need some preparation.

First we have to distinguish between the vocabulary of (ordered) fields 
$\tau_{+,\times}$ 
($\tau_{o,+,\times}$) 
where addition and multiplication
are binary functions, and $-x$ and $\frac{1}{x}$ are unary functions,
and
$\tau_{field}$ 
($\tau_{ofield}$) 
where addition and multiplication are ternary relations, and there are no
symbols for $-x$ and $\frac{1}{x}$.
The distinction is necessary, because the notion of substructure is different.

Next we need the a special case
of Tarski's Theorem for universal formulas proven in every
textbook on model theory, e.g., \cite{bk:Hodges93}.
\begin{lemma}
\label{le:substructures}
Let $\cF$ be a field and $\cF_0$ be a subfield.
Let $\theta \in \FOL(\tau_{+,\times})$ be a universal sentence, and $\cF \models \theta$
Then $\cF_0 \models \theta$.
\\
The same also holds for ordered fields.
\end{lemma}

We also need the characterizations of the first order theory of the real and complex numbers.

Let $ACF_0 \subseteq \FOL(\tau_{+,\times})$ be theory of algebraically closed
fields of characteristic $0$, and
Let $RCF \subseteq \FOL(\tau_{o+,\times})$ be theory of real closed ordered fields.

\begin{proposition}
\label{pr:tarski}
$ACF_0$ and $RCF$
are complete, i.e., for every 
$\tau_{field}$-sentence 
respectively $\tau_{o,+,\times}$-sentence,
$\theta$ we have that 
$ACF_0 \models \theta$ or
$ACF_0 \models \neg\theta$ but not both,
and the same for
a $\tau_{o,+,\times}$-sentence and $RCF$.
\end{proposition}

\begin{lemma}
\label{le:universal}
\begin{enumerate}
\item
Let $F$ be a set of 
$\tau_{+,\times}$-sentences
such that all its models are fields of characteristic $0$, and $F$ is consistent with $ACF_0$.
If $\theta$ is a universal
$\tau_{+,\times}$-sentence,
then
$$
F \models \theta
\mbox{ iff }
ACF_0 \models \theta.
$$
\item
Let $F$ be a set of 
$\tau_{o,+,\times}$-sentences
such that all its models are ordered fields, and $F$ is consistent with $RCF$.
If $\theta$ is a universal
$\tau_{o,+,\times}$-sentence,
then
$$
F \models \theta
\mbox{ iff }
RCF \models \theta.
$$
\end{enumerate}
\end{lemma}
\begin{proof}
(i):
As $ACF_0$ is complete and $F$ is consistent with $ACF_0$ we have that
$ACF_0 \models F$.
Let $F \models \theta$. Then also $ACF_0 \models \theta$.
\\
Conversely, assume $ACF_0 \models \theta$. Now we use that $\theta$ is universal.
By Lemma \ref{le:substructures}, 
$\theta$ holds in every subfield $\cF$ of an algebraically closed fields of characteristic $0$.
By a classical theorem of Algebra, \cite{ar:steinitz},
every field of characteristic $0$ has an algebraically closed extension which satisfies
$ACF_0$. Hence $T \models \theta$.
\\
(ii): The proof is similar, using real closures instead.
\hfill $\Box$
\end{proof}
\fi 

The reader can easily verify the following:

\begin{proposition}
\label{qfree-transductions}
The formulas in the definition of the translation scheme 
$$
PP = \langle
\phi_{points},
\phi_{lines},
\phi_{\in},
\phi_{Eq},
\phi_{Or},
\phi_{An},
\phi_{Be}
\rangle
$$
can be written in the vocabulary $\tau_{o,+, \times}$ as quantifier-free formulas
$$
QP = \langle
\psi_{points},
\psi_{lines},
\psi_{\in},
\psi_{Eq},
\psi_{Or},
\psi_{An},
\psi_{Be}
\rangle
$$
In particular if 
$\theta \in \FOL_{hilbert}$
is a universal formula, 
then $ \hat{\theta}= QP^{\sharp}(\theta) \in \FOL(\tau_{o,+,\times}$
is a universal formula. 
\end{proposition}

The general form of Proposition \ref{pr:universal} can now be stated as follows:
\begin{theorem}
\label{th:universal}
$T \subseteq \FOL(\tau_{hilbert})$ be a geometrical theory such  that
\begin{enumerate}[(i)]
\item
$T \models T_{pappus}$.
\item
The set of formulas $$
F_T = \{ \phi \in \FOL(\tau_{o,+,\times}): \phi =  PP^{\sharp}(\psi) \wedge \psi \in T \}
$$ 
is consistent with $RCF$.
\item 
For every $\Pi$ with $\Pi \models T$
$PP^*(RR^*(\Pi))$ is isomorphic to $\Pi$.
\end{enumerate}
Then
the universal consequences of $T$ are decidable.
\\
The analogous statement also holds for
$T \subseteq \FOL(\tau_{wu})$ and
$T \subseteq \FOL(\tau_{origami})$ where $RCF$ is replaced by $ACF_0$.
\end{theorem}
\begin{proof}[of Theorem {\ref{th:universal}}]
Let $T$ be as required and $\theta$ be universal.
We want to show that $T \models \theta$ iff $RFC \models QP^{\sharp}(\theta)$. 
The latter can be decided using
the Theorems \ref{th:seidenberg}. 

We have: $T \models \theta$ iff for every plane $\Pi$ with $\Pi \models T$ also $\Pi \models \theta$.
\\
By Theorem \ref{th:szmielew} $\Pi$ is isomorphic $PP^*(RR^*(\Pi))$  and also to $QP^*(RR^*(\Pi))$.
By Theorem \ref{pr:Fund} $\Pi \models \theta$ iff $(RR^*(\Pi) \models QP^{\sharp}(\theta)$.
By the definition of $QP$ the formula $QP^{\sharp}(\theta)$ is universal.
Therefore $(RR^*(\Pi) \models QP^{\sharp}(\theta)$ iff $RCF \models QP^{\sharp}(\theta)$
by Lemma \ref{pr:f-universal}.

In the case of fields rather than ordered fields, we 
show that $T \models \theta$ iff $ACF_0 \models QP^{\sharp}(\theta)$ 
which can be decided using
Theorem \ref{th:tarski}. 
\hfill $\Box$
\end{proof}
Proposition \ref{pr:universal} now follows easily
using Theorems \ref{th:seidenberg} and \ref{th:tarski}.

We also get:
\begin{corollary}
The universal consequences of $T_{a-origami}$ formulated as formulas in $\FOL(\tau_{wu})$
are decidable.
\end{corollary}
\begin{proof}
This follows from the characterization of the fields corresponding to $T_{origami}$
as the Vieta fields (Theorem \ref{th:origami}).
\end{proof}

More decidability for the universal consequences can be obtained 
from axiomatizations of geometrical constructions using more than just ruler and compass,
cf.  \cite{pambuccian2008axiomatizing}.

\section{Conclusions}
\label{se:conclu}
We have discussed the decidability of the consequence problem for various axiomatizations of Euclidean
geometry. The purpose of the paper was to make the metamathematical methods discussed
in \cite{bk:Schwabhaeuser1983}  and in \cite{balbiani2007logical}
more accessible to the research communities of symbolic computation
and automated theorem proving.  In particular, we wanted to draw attention to Ziegler's Theorem
\ref{th:Ziegler}, and spell out in detail what is needed to draw its consequences for geometrical theories. 
We have also listed some open problems concerning the decidability of theories of fields
if restricted to fragments of first order logic such as $\fU, \fE, \fH$.

In writing this expository paper we also included
new applications of these methods to Wu's orthogonal geometry and to the geometry
of paper folding Origami. These results, both undecidability of first order consequences and
decidability of universal consequences, can be easily extended to theories of geometric constructions
going beyond ruler and compass or paper folding, cf. \cite{bk:Hartshorne2000,pambuccian2008axiomatizing}.

From a complexity point of view, we see that the consequence problem for first order
formulas is either undecidable or, in the case of Tarski's decidability results,
prohibitively difficult. We have also shown that in the cases discussed, the consequence
problem for universally quantified formulas is decidable, possibly in nondeterministic polynomial time.
What is left open, and remains a challenge for future research, is the decidability question
for existential and $\forall \exists$-Horn formulas $\fE$ and $\fH$.

\small
\subsubsection*{Acknowledgements}
This paper has its origin in my lecture notes on automated theorem proving \cite{JAM-ATP},
developed in the last 15 years. 
I was motivated to develop this material further, when I prepared a lecture
on P. Bernays and the foundations of geometry, which I gave
at the occasion of the unveiling in summer 2017 of a plaque at the house where P. Bernays 
used to live in G\"ottingen, before going into forced exile in 1933.
P. Bernays edited Hilbert's \cite{hilbert1902foundations} from the 5th (1922) till the 10th edition (1967),
see also \cite{hilbert1971foundations,hilbert2013grundlagen}.
I am indebted to R. Kahle, who invited me to give this lecture.
Without this invitation this paper would not have been written.
I would also like to thank L. Kovacs for her patience and flexibility concerning the
deadline for submitting this paper to the special issue
on {\em Formalization of Geometry and Reasoning} of the {\em Annals of Mathematics and Artificial Intelligence}.
Special thanks are due to the anonymous referees and to J. Baldwin for critical remarks
and suggestions, as well as for pointing out various imprecisions, which I hope were
all corrected.

I was lucky enough to know P. Bernays personally, as well as some other 
pioneers of the modern foundations of geometry, among them
R. Baer, 
H. Lenz, 
W. Rautenberg,
W. Schwabh\"auser, 
W. Szmielew and
A. Tarski. 
I dedicate this paper to them, and to  my wonderful teacher of descriptive geometry, 
M. Herter, at the Gymnasium Freudenberg
in Zurich, Switzerland.
Blessed be their memory.


\bibliographystyle{alpha}      
\bibliography{r-ref}   
\end{document}